\pdfoutput=1

\newif\iflongversion

\longversiontrue 

\iflongversion
\documentclass[12pt,draftcls,onecolumn]{IEEEtran}
\else 
\documentclass{IEEEtran}
\fi

\usepackage{graphicx}
\usepackage{amsmath,amsthm,amsfonts,amssymb}
\usepackage{floatrow}
\usepackage{caption}
\usepackage{placeins}
\usepackage{makecell}
\usepackage{subcaption}
\usepackage{multirow}
\usepackage{array}
\usepackage{mathtools}
\usepackage{color}
\usepackage{bbm}
\usepackage{wasysym}
\usepackage[dvipsnames]{xcolor}
\usepackage{tikz}
\usepackage{enumerate}
\usepackage{epstopdf} 
\usepackage{cite}
\usepackage{setspace}

\usetikzlibrary{arrows,positioning,shapes,backgrounds,fit}  
\usetikzlibrary{positioning}

\usepackage{blkarray}

\DeclareMathOperator*{\argmax}{\arg\max}   

\newtheorem{theorem}{Theorem}
\newtheorem{definition}{Definition}

\newtheorem{assumption}{Assumption}
\newtheorem{lemma}{Lemma}

\newtheorem{remark}{Remark}

\newcommand{\red}[1]{\textcolor{black}{#1}}

\begin{document}
	\title{A Unified Approach to Dynamic Decision Problems with Asymmetric Information -\\  Part I: Non-Strategic Agents}   
	\author{Hamidreza Tavafoghi, Yi Ouyang, and Demosthenis Teneketzis}         
	
	\date{November 23, 2018} 
	\maketitle
	{\let\thefootnote\relax\footnote{ A preliminary version of this paper will appear in
			the Proceeding of the 57th IEEE Conference on Decision and Control
			(CDC), Miami Beach, FL, December 2018 \cite{CDC18}.\\ H. Tavafoghi is with the Department of Mechanical Engineering at the University of California, Berkeley (e-mail:
			tavaf@berkeley.edu). Y. Ouyang is with Preferred Networks America, Inc. (e-mail: ouyangyi@preferred-america.com). D. Teneketzis is with the Department of Electrical Engineering and Computer Science at the University of Michigan, Ann Arbor (e-mail:
			teneket@umich.edu)\\
			This work was supported in part by the NSF grants CNS-1238962, CCF-1111061, ARO-MURI grant W911NF-13-1-0421, and ARO grant W911NF-17-1-0232.\\
		}}
		\begin{abstract}
				  We study a general class of dynamic multi-agent decision problems with asymmetric information and non-strategic agents, which includes dynamic teams as a special case.  When agents are non-strategic, an agent's strategy is known to the other agents. Nevertheless, the agents' strategy choices and beliefs are interdependent over times, a phenomenon known as \textit{signaling}.
				   We introduce the \red{notions of \textit{private information} that effectively compresses} the agents' information in a mutually consistent manner. Based on the notions of sufficient information,  we propose an information state for each agent that is sufficient for decision making purposes. We present instances of dynamic multi-agent decision problems where we can determine an information state with a time-invariant domain for each agent. Furthermore, we present a generalization of \textit{the policy-independence property of belief} in Partially Observed Markov Decision Processes (POMDP) to dynamic multi-agent decision problems.  
				   Within the context of dynamic teams with asymmetric information, the proposed set of information states leads to a sequential decomposition that decouples the interdependence between the agents' strategies and beliefs over time, and enables us to formulate a dynamic program to determine a globally optimal policy via backward induction.

		\end{abstract}

		\vspace*{-5pt}
		\section{Introduction}\label{sec:introductionG}

\subsection{Background and Motivation}

\red{Dynamic multi-agent decision problems with asymmetric information have been used to model many situations arising in engineering, economic, and socio-technological applications.}
In these applications many decision makers/agents interact with each other as well as with a dynamic system. They make private imperfect observations over time, and influence the evolution of the dynamic system through their actions that are determined by their strategies.
An agent's strategy is defined as a decision rule that the agent uses to choose his action at each time based on his realized information at that time.

In this paper, we study a general class of dynamic decision problems with \textit{non-strategic agents}. We say an agent is non-strategic if his strategy (not his specific action) is known to the other agents. In a companion paper \cite{game} we study dynamic decision problems with \textit{strategic agents} where an agent's strategy is his private information and not known to the other agents.

We consider an environment with controlled Markovian dynamics, where, given the agents' actions at every time, the system state at the next time is a stochastic function of the current system state. The instantaneous utility of each agent depends on the agents' joint actions as well as the system state. At every time, each agent makes a private noisy observation that depends on the current system state and past actions of all agents in the system.
Therefore, agents have asymmetric and imperfect information about the system history. 
Moreover, each agent's information depends on other agents' past actions and strategies; this phenomenon is known as \textit{signaling} in the control theory literature. In such problems, the agents' decisions and information are coupled and interdependent over time because (i) an agent's utility depends on the other agents' actions, (ii) the evolution of the system state depends, in general, on all the agents' actions, (iii) each agent has imperfect and asymmetric information about the system history, and (iv) at every time an agent's information depends, in general, on the agents' (including himself)  past actions and strategies.

There are two main challenges in the study of dynamic multi-agent decision problems with asymmetric information. First, because of the coupling and interdependence among the agents' decisions and information over time, we need to determine the agents' strategies simultaneously for all times. Second, as the agents acquire more information over time, the domains of their strategies grow.

In this paper, we propose a general approach for the study of dynamic decision problems with non-strategic agents and address these two challenges. 
We propose the notion of \textit{sufficient information} and provide a set of conditions sufficient to characterize a compression of the agents' private and common information in a mutually consistent manner over time. 
We show that such a compression results in an \textit{information state} for each agent's decision making problem. 
We show that restriction to the set of strategies based on this information state entails no loss of generality in dynamic decision problems with non-strategic agents.

We identify specific instances of dynamic decision problems where we can discover a set of information states for the agents that have time-invariant domain.
Within the context of dynamic teams, we further demonstrate that the notion of \red{sufficient information} leads to a sequential decomposition of dynamic teams. This sequential decomposition results in a dynamic program the solution of which determines the agents' \textit{globally optimal strategies}.

\vspace*{-5pt}

\subsection{Related Literature}
The Partially Observed Markov Decision Processes (POMDPs), \textit{i.e.} centralized stochastic control problems, present the simplest form of dynamic decision problems with single agent 
	\cite{kumar1986stochastic,Bertsekas:1995}. 
	To analyze and identify properties of optimal strategies in POMDPs the notion of \textit{information state} is introduced as the agent's belief about the current system state conditioned on his information history. The information state provides a way to compress the agent's information over time that is sufficient for the decision-making purposes. 
When the agent has perfect recall, this information state is independent of the agent's strategies over time; this result is known as the \textit{policy-independence belief} property \cite{kumar1986stochastic}.

Dynamic multi-agent decision problems with non-strategic agents are considerably more difficult compared to their centralized counterparts. This is because, due to signaling, they are (in general) non-convex functional optimization problems (see \cite{mahajan_martins_yuksel,kulkarni2015optimizer,lessard2016convexity,yuksel2016convex}). The difficulties present in these problems were first illustrated by Witsenhausen \cite{Witsenhausen:1968}, who showed that in a simple dynamic team problem with Gaussian primitive random variables and quadratic cost function where signaling occurs, linear strategies are suboptimal (contrary to the corresponding centralized problem where linear strategies are optimal). Subsequently, many researchers investigated control problems with various specific information structures such as: partially nested (\cite{Chu:1972,LamperskiDoyle:2011,LessardNayyar:2013,ShahParrilo:2013, Nayyar_Lessard_2015,Lessard_Lall_2015} and references therein), stochastic nested \cite{Yuksel:2009}, randomized partially nested \cite{Ouyang_Asghari_Nayyar:2017}, delayed sharing (\cite{LamperskiDoyle:2011,nayyar2011optimal,witsenhausen1971separation,varaiya1978delayed} and references therein), information structures possessing the \textit{i-partition} property or the \textit{ s-partition} property \cite{yoshikawa1978decomposition}, the quadratic invariance property  \cite{RotkowitzLall:2006}, and the substitutability property \cite{asghari2016dynamic}.

Currently, there are three approaches to the analysis of dynamic multi-agent decision problems with non-strategic agents: the \textit{agent-by-agent} approach \cite{ho1980team}, the \text{designer's} approach \cite{witsenhausen1973standard}, and the \textit{common information} approach \cite{nayyar2013decentralized}. We provide a brief discussion of these approaches here. We discuss them in details in Section \ref{sec:disc:c}, where we compare them with the sufficient information approach we present in this paper and show that our approach is distinctly different from them. 

The agent-by-agent approach \cite{ho1980team}, is an iterative method. At each iteration, we pick an agent and fix the strategy of all agents except that agent, and determine the best response for that agent and update his strategy accordingly. We proceed in a round robin fashion among the agents until a fixed point is reached, that is, when no agent can improve his performance by unilaterally changing his strategy. The designer's approach \cite{witsenhausen1973standard}, considers the decision problem from the point of view of a designer who knows the system model and the probability distribution of the primitive random variables, and chooses the control strategies for all agents without having an information about the realization of the primitive random variables. The common information approach \cite{nayyar2013decentralized}, assumes that at each time all agents possess private information and share some common information; it uses the common information to coordinate the agents' strategies \red{sequentially over time.}

\vspace*{-5pt}

\subsection{Contribution}
We develop a general methodology for the study and analysis of dynamic decision problems with asymmetric information and non-strategic agents. Our model includes problems with  non-classical information structures \cite{witsenhausen1971separation} where signaling is present. We propose an approach that effectively compresses the agents' private and common information in a mutually consistent manner. As a result, we offer a set of information states for the agents which are sufficient for decision making purposes. We characterize special instances where we can identify an information state with a time-invariant domain. Based on the proposed information state, we provide a sequential decomposition of dynamic teams over time. We show that the methodology developed in this paper generalizes the existing results for dynamic teams with non-classical information structure. 
Our results in this paper, along those appearing in the companion paper \cite{game} present a set of information states sufficient for decision making in strategic and non-strategic settings. Therefore, we provide a unified approach to decision making problems that can be used to study dynamic games and dynamic teams as well as dynamic games among teams of agents.

\vspace*{-7pt}

\subsection{Organization}

The rest of the paper is organized as follows. In Section \ref{sec:model}, we describe the model and present few examples. In Section \ref{sec:equilibrium}, we discuss the main challenges that are present in dynamic multi-agent decision problems with non-strategic agents. We present the sufficient information approach in Section \ref{sec:sufficient}. We present the main results of the paper in Section \ref{sec:nonstrategic}. We discuss an open problem associated with the sufficient information approach in Section \ref{sec:disc:a}. In Section \ref{sec:disc:c}, we compare the sufficient information approach with the existing approaches in the literature.  We provide a generalization of the sufficient information approach in Section \ref{sec:disc:b}. We present an extension of our results to infinite-horizon dynamic multi-agent decision problems with non-strategic agents in Section \ref{sec:infinite}. We conclude in Section \ref{sec:conclsion}. The proofs of all the theorems and lemmas appear in the Appendix.

\vspace*{-7pt}

\subsection*{Notation}
Random variables are denoted by upper case letters, their realizations by the corresponding lower case letters.
In general, subscripts are used as time index while superscripts are used to index agents.
For $t_1\hspace*{-2pt}\leq\hspace*{-2pt} t_2$, $X_{t_1:t_2}$ (resp. $f_{t_1:t_2}(\cdot)$) is the short hand notation for the random variables $(X_{t_1},\hspace*{-1pt}X_{t_1+1},...,\hspace*{-1pt}X_{t_2})$ (resp.  functions $(f_{t_1}(\cdot),\dots,\hspace*{-1pt}f_{t_2}(\cdot))$).
When we consider a sequence of random variables (resp. functions) for all time, we drop the subscript and use $X$ to denote $X_{1:T}$ (resp. $f(\cdot)$ to denote $f_{1:T}(\cdot)$).
For random variables $X^1_t,\dots,\hspace*{-1pt}X^N_t$ (resp. functions $f^1_t(\cdot),\dots,\hspace*{-1pt}f^N_t(\cdot)$), we use $X_t\hspace*{-2pt}:=\hspace*{-2pt}(X^1_t,\dots,\hspace*{-1pt}X^N_t)$ (resp. $f_t(\cdot)\hspace*{-2pt}:=\hspace*{-2pt}(f^1_t(\cdot),\dots,\hspace*{-1pt}f^N_t(\cdot))$) to denote the vector of the set of random variables (resp. functions) at $t$, and $X^{-n}_t\hspace*{-2pt}:=\hspace*{-2pt}(X^1_t,\dots,\hspace*{-1pt}X^{n-1}_t,\hspace*{-1pt}X^{n+1}_t,\dots,\hspace*{-1pt}X^N_t)$ (resp. $f^{-n}_t(\cdot)\hspace*{-2pt}:=\hspace*{-2pt}(f^1_t(\cdot),\dots,\hspace*{-1pt} f^{n-1}_t(\cdot),\hspace*{-1pt}f^{n+1}_t(\cdot),\dots,\hspace*{-1pt}f^N_t(\cdot))$) to denote all random variables (resp. functions) at $t$ except that of the agent indexed by $n$.
$\mathbb{P}(\cdot)$ and $\mathbb{E}(\cdot)$ denote the probability and expectation of an event and a random variable, respectively.
For a set $\mathcal{X}$, $\Delta(\mathcal{X})$ denotes the set of all beliefs/distributions on $\mathcal{X}$.
For random variables $X,\hspace*{-1pt}Y$ with realizations $x,\hspace*{-1pt}y$, $\mathbb{P}(x|y)\hspace*{-2pt} := \hspace*{-2pt}\mathbb{P}(X\hspace*{-2pt}=\hspace*{-2pt}x|Y\hspace*{-2pt}=\hspace*{-2pt}y)$ and $\mathbb{E}(X|y) \hspace*{-2pt}:=\hspace*{-2pt} \mathbb{E}(X|Y\hspace*{-2pt}=\hspace*{-2pt}y)$.
For a strategy $g$ and a belief (probability distribution) $\pi$, we use $\mathbb{P}^g_{\pi}(\cdot)$ (resp. $\mathbb{E}^g_{\pi}(\cdot)$) to indicate that the probability (resp. expectation) depends on the choice of $g$ and $\pi$. We use $\mathbf{1}_{\{X=x\}}$ to denote the indicator function for event $X\hspace*{-2pt}=\hspace*{-2pt}x$. For sets $A$ and $B$ we use $A\backslash B$ to denote all elements in set $A$ that are not in set $B$. For random variables $X$ and $Y$ we write $X\hspace*{-2pt}\stackrel{dist.}{=}\hspace*{-2pt}Y$ when $X$ and $Y$ have an identical probability distribution. 
		\vspace*{-5pt}
\section{Model}\label{sec:model}\vspace*{-2pt}
\textit{1) System dynamics:} Consider $N$ non-strategic agents who live in a dynamic Markovian world over a horizon $\mathcal{T}\hspace*{-3pt}:=\hspace*{-5pt}\{\hspace*{-1pt}1,\hspace*{-1pt}2,...,\hspace*{-1pt}T\}$, $T\hspace*{-2pt}<\hspace*{-2pt}\infty$. Let $X_t\hspace*{-2pt}\in\hspace*{-2pt}\mathcal{X}_t$ denote the state of the world at $t\hspace*{-2pt}\in\hspace*{-2pt}\mathcal{T}$. At time $t$, each agent, indexed by $i\hspace*{-2pt}\in\hspace*{-2pt} \mathcal{N}\hspace*{-2pt}:=\hspace*{-2pt}\{1,\hspace*{-1pt}2,...,\hspace*{-1pt}N\}$, chooses an action $a^i_t\hspace*{-2pt}\in\hspace*{-2pt}\mathcal{A}^i_t$, where  $\mathcal{A}^i_t$ denotes the set of available actions to him at $t$. Given the collective action profile $A_t\hspace*{-2pt}:=\hspace*{-2pt}(A_t^1,...,\hspace*{-1pt}A_t^N)$, the state of the world evolves according to the following stochastic dynamic equation,
\begin{align}
X_{t+1}=f_t(X_t,A_t,W_t^x), \label{eq:systemdynamic1} \vspace*{-2pt}
\end{align} 
where $W_{1:T-1}^x$ is a sequence of independent random variables. The initial state $X_1$ is a random variable that has a probability distribution $\eta\in\Delta(\mathcal{X}_1)$ with full support.

At every time $t\in\mathcal{T}$, before taking an action, agent $i$ receives a noisy private observation $Y_t^i\in\mathcal{Y}_t^i$ of the current state of the world $X_t$ and the action profile $A_{t-1}$, given by
\begin{align}
Y_t^i=O_t^i(X_t,A_{t-1},W_t^i), \label{eq:systemdynamic2} \vspace*{-2pt}
\end{align} 
where $W_{1:T}^i$, $i\in\mathcal{N}$, are sequences of independent random variables. Moreover, at every $t\in\mathcal{T}$, all agents receive a common observation $Z_t\in\mathcal{Z}_t$ of the current state of the world $X_t$ and the action profile $A_{t-1}$, given by
\begin{align}
Z_t=O_t^c(X_t,A_{t-1},W_t^c), \label{eq:systemdynamic3} \vspace*{-2pt}
\end{align} 
where $W_{1:T}^c$, is a sequence of independent random variables. We note that the agents' actions $A_{t-1}$ is commonly observable at $t$ if $A_{t-1}\subseteq Z_t$.
We assume that the random variables $X_1$, $W_{1:T-1}^x$, $W_{1:T}^c$, and $W_{1:T}^i$, $i\in\mathcal{N}$ are mutually independent. 

\vspace{5pt}

\textit{2) Information structure:} Let $H_t\in\mathcal{H}_t$ denote the aggregate information of all agents at time $t$. Assuming that agents have perfect recall, we have $H_t=\{Z_{1:t},Y_{1:t}^{1:N},A_{1:t-1}^{1:N}\}$, \textit{i.e.} $H_t$ denotes the set of all agents' past observations and actions. The set of all possible realizations of the agents' aggregate information is given by $\mathcal{H}_t:=\prod_{\tau\leq t}\mathcal{Z}_\tau\times\prod_{i\in\mathcal{N}}\prod_{\tau\leq t}\mathcal{Y}_\tau^i\times \prod_{i\in\mathcal{N}}\prod_{\tau< t}\mathcal{A}_\tau^i$. 

At time $t\hspace*{-2pt}\in\hspace*{-2pt}\mathcal{T}$, the aggregate information $H_t$ is not fully known to all agents; each agent may have asymmetric information about $H_t$. Let $C_t\hspace*{-2pt}:=\hspace*{-2pt}\{\hspace*{-1pt}Z_{1:t}\}\hspace*{-2pt}\in\hspace*{-1pt}\mathcal{C}_t$ denote the agents' common information  about $H_t$ and $P_t^i\hspace*{-2pt}:=\hspace*{-2pt}\{\hspace*{-1pt}Y_{1:t}^i,\hspace*{-1pt}A_{1:t-1}^i\}\backslash C_t\hspace*{-2pt}\in\hspace*{-1pt}\mathcal{P}_t^i$ denote agent $i$'s private information about $H_t$, where $\mathcal{P}_t^i$ and $\mathcal{C}_t$ denote the set of all possible realizations of agent $i$'s private and common information at $t$, respectively. 
In this paper, we discuss several instances of information structures  that can be captured as special cases of our general model.

\vspace{5pt}

\textit{3) Strategies and Utilities:} Let $H_t^i\hspace*{-2pt}:=\hspace*{-2pt}\{C_t,\hspace*{-1pt}P_t^i\}\hspace*{-2pt}\in\hspace*{-1pt} \mathcal{H}_t^i$ denote the information available to agent $i$ at $t$, where $\mathcal{H}_t^i$ denote the set of all possible realizations of agent $i$'s information at $t$. Agent $i$'s \textit{strategy} $g^i:=\{g_t^i, t\in\mathcal{T}\}$, is defined as a sequence of mappings $g_t^i:\mathcal{H}_t^i\rightarrow \Delta (\mathcal{A}_t^i)$, $t\hspace*{-2pt}\in\hspace*{-1pt}\mathcal{T}$, that determine agent $i$'s action $A_t^i$ for every realization $h_t^i\hspace*{-2pt}\in\hspace*{-1pt}\mathcal{H}_t^i$ of \red{his} history  at $t\hspace*{-2pt}\in\hspace*{-1pt}\mathcal{T}$. 

Agent $i$'s instantaneous utility at $t$ depends on the state of the world $X_t$ and the collective action profile $A_t$ and is given by $u_t^i(X_t,A_t)$. Therefore, agent $i$'s total utility over the horizon $\mathcal{T}$ is given as
\begin{align}
U^i(X_{1:T},A_{1:T}):=\sum_{t\in\mathcal{T}}u_t^i(X_t,A_t). \label{eq:totalutility}
\end{align}

We assume that agents are \textit{non-strategic}. That is, each agent's, say $i$'s, $i\hspace*{-2pt}\in\hspace*{-2pt}\mathcal{N}$, strategy choice $g^i$ is known to other agents. 
We note that these non-strategic agents may have different utilities over time.
Therefore, the model includes a team of agents sharing the same utilities (see Sections \ref{sec:nonstrategic}) as well as agents with general non-identical utilities.
In \cite{game} we build on our results in this paper to study dynamic decision problems with strategic agents where an agent may deviate privately from the commonly believed strategy, and gain by misleading the other agents.


To avoid measure-theoretic technical difficulties and for clarity and convenience of exposition, we assume that all the random variables take values in finite sets.
\begin{assumption}\label{assump:finite}(Finite game)
	The sets $\mathcal{X}_t$, $\mathcal{Z}_t$, $\mathcal{Y}_t^i$, $\mathcal{A}_t^i$, $i\in\mathcal{N}$, $t\in\mathcal{T}$, are finite.
\end{assumption}

\vspace*{3pt}
\noindent\textbf{Special Cases:}
\vspace*{3pt}

We present several instances of dynamic decision problems with asymmetric information that are special cases of the general model described above.  

\vspace{3pt}

\textit{1) Real-time source coding-decoding \cite{witsenhausen1979structure}:} Consider a data source that generates a random sequence $\{\hspace*{-1pt}X_1\hspace*{-1pt},\hspace*{-1pt}...,\hspace*{-1pt}X_T\hspace*{-1pt}\}$ that is $k$-th order Markov, \textit{i.e.} for every sequence of realizations  $x_{1:T}$, $\mathbb{P}\{\hspace*{-1pt}X_{t+k:T}\hspace*{-2pt}=\hspace*{-2pt}x_{t+k:T}|x_{1:t+k-1}\hspace*{-1pt}\}\hspace*{-2pt}=\hspace*{-2pt}\mathbb{P}\{\hspace*{-1pt}X_{t+k:T}\hspace*{-2pt}=\hspace*{-2pt}x_{t+k:T}|x_{t:t+k-1}\hspace*{-1pt}\}$ for $t\hspace*{-2pt}\leq \hspace*{-2pt}T-k$. There exists an encoder (agent $1$) who observes $X_t$ at every time $t$; the encoder has perfect recall.
At every time $t$, based on his available data $\{\hspace*{-1pt}X_1\hspace*{-1pt},\hspace*{-1pt}...,\hspace*{-1pt}X_t\hspace*{-1pt}\}$, the encoder transmits a signal $M_t\hspace*{-2pt}\in\hspace*{-2pt}\mathcal{M}_t$ through a noiseless channel to a decoder (agent $2$), where $\mathcal{M}_t$ denotes the transmission alphabet.  At the receiving end, at every time $t$, the decoder wants to estimate the value of $X_{t-1-\delta}$ (with delay $\delta$) as $\hat{X}_{t-1-\delta}$ based on his available data $M_{1:t-1}$; we  assume that the decoder has perfect recall. The encoder and decoder choose their \red{joint} coding-decoding policy so as to minimize the expected total distortion function given by $\sum_{t=2+\delta}^T d_t(X_t,\hat{X}_t)$, where $d_t(\cdot,\cdot)$ denotes the instantaneous distortion function.
To capture the above-described model within the context of our model, we need to define an augmented system state $\tilde{X}_t$ that includes the last $\max(k,\delta+1)$ states realizations as $\tilde{X}_t\hspace*{-2pt}:=\hspace*{-2pt}\{\hspace*{-1pt}X_{t-\max(k,\delta+1)+1}\hspace*{-1pt},\hspace*{-1pt}...,\hspace*{-1pt}X_t\hspace*{-1pt}\}$. Moreover, the encoder's (agent $1$'s) observation is given by $Y_t^1=O_t^1(\tilde{X}_t,A_{t-1})=X_t$ and the decoder's (agent $2$'s) observation is given by $Y_t^2\hspace*{-2pt}=\hspace*{-2pt}O_t^2(\tilde{X}_t,\hspace*{-1pt}A_{t-1})\hspace*{-2pt}=\hspace*{-2pt}M_{t-1}$, where $(A_t^1,\hspace*{-1pt}A_t^2)\hspace*{-2pt}=\hspace*{-2pt}(M_t,\hspace*{-1pt}\hat{X}_{t-1-\delta})$. The encoder's and decoder's instantaneous utility are given by a distortion function $u_t^{\text{team}}(\tilde{X}_t\hspace*{-1pt},\hspace*{-1pt}A_t\hspace*{-2pt})=\hspace*{-2pt}d_t(X_{t-1-\delta},\hspace*{-1pt}\hat{X}_{t-1-\delta})$.

\vspace{3pt}

\tikzstyle{block} = [draw, rectangle, 
minimum height=2em, minimum width=4em]
\tikzstyle{sum} = [draw, fill=blue!20, circle, node distance=1cm]
\tikzstyle{input} = [coordinate]
\tikzstyle{output} = [coordinate]
\tikzstyle{pinstyle} = [pin edge={to-,thin,black}]


\textit{2) Delayed sharing information structure \cite{witsenhausen1971separation,varaiya1978delayed,kurtaran1979corrections,nayyar2011optimal}:} Consider a $N$-agent decision problem where agents observe each others' observations and actions with $d$-step delay. 
We note that in \red{our} model we assume that the agents' common observation $Z_t$ at $t$ is only a function of $X_t$ and and $A_{t-1}$. Therefore, to describe the decision problem with delayed sharing information structure within the context of our model we need to augment our state space to include the agents' last $d$ observations and actions as part of the augmented state. Define $\tilde{X}_t:=\{X_t,M^1_t,M^2_t,...,M^d_t\}$ as the augmented system state where $M_t^i\hspace*{-2pt}:=\hspace*{-2pt}\{A_{t-i},Y_{t-i}\}\in\mathcal{A}_{t-i}\times\mathcal{Y}_{t-i}$, $i\hspace*{-2pt}\in\hspace*{-2pt}\mathcal{N}$; that is, $M_t^i$ serves as a temporal memory for the agents' observations $Y_{t-i}$ \red{and actions $A_{t-i}$} at $t\hspace*{-2pt}-\hspace*{-2pt}i$. Then, we have $\tilde{X}_{t+1}\hspace*{-2pt}=\hspace*{-2pt}\{X_{t+1},M_{t+1}^1,M_{t+1}^2,...,M_{t+1}^d\}\hspace*{-2pt}=\hspace*{-2pt}\{f_t(X_t,A_t,W_t^x),(Y_t,A_t),M_t^1,...,M_t^{d-1}\}$ and $Z_t\hspace*{-2pt}=\hspace*{-2pt}\{M_t^d\}\hspace*{-2pt}=\hspace*{-2pt}\{Y_{t-d},A_{t-d}\}$.

\vspace{3pt}

\textit{3) Real-time multi-terminal communication \cite{nayyar2011structure}:} Consider a real-time communication system with two encoders (agents $1$ and $2$) and one receiver (agent $3$). The two encoders make distinct observations $X_t^1$ and $X_t^2$ of a Markov source. The encoders' observation are conditionally independent Markov chains. That is, there is an unobserved random variable variable $R$ such that $\mathbb{P}\{\hspace*{-1pt}X_1^1\hspace*{-1pt},\hspace*{-1pt}X_1^2\hspace*{-1pt},\hspace*{-1pt}R\hspace*{-1pt}\}\hspace*{-2pt}=\hspace*{-2pt}\mathbb{P}\{\hspace*{-1pt}X_1^1|R\}\mathbb{P}\{\hspace*{-1pt}X_1^2|R\}\mathbb{P}\{R\}$, and
\begin{align*}\mathbb{P}\{X_{t+1}^1,X_{t+1}^2|X_t^1,X_t^2,R\}\hspace*{-2pt}=\hspace*{-2pt}\mathbb{P}\{X_{t+1}^1|X_t^1,R\}\mathbb{P}\{X_{t+1}^2|X_t^2,R\}.\end{align*}


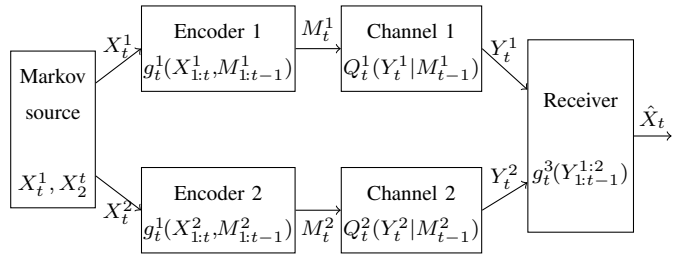
\begin{center}{\fontsize{8}{9}\selectfont
	\begin{tikzpicture}[scale=0.2]
	\hspace*{-9pt}\node [block] (source) {$\begin{array}{c}\hspace*{-7pt}\text{Markov}\hspace*{-7pt}\\\hspace*{-7pt}\text{source}\hspace*{-7pt}\\ \\\hspace*{-7pt}X_t^1,X_2^t\hspace*{-7pt}\end{array}$};
	\node [block] (encoder1) [above right=-0.55cm and 0.6cm of source] {$\begin{array}{c}\text{Encoder 1}\\\hspace*{-7pt}g_t^1\hspace*{-1pt}(X_{1\hspace*{-1pt}:t}^1,\hspace*{-2pt}M_{1\hspace*{-1pt}:t-1}^1)\hspace*{-7pt}\end{array}$};
	\node [block] (encoder2) [below right=-0.55cm and 0.6cm of source] {$\begin{array}{c}\text{Encoder 2}\\\hspace*{-7pt}g_t^1\hspace*{-1pt}(X_{1\hspace*{-1pt}:t}^2,\hspace*{-2pt}M_{1\hspace*{-1pt}:t-1}^2)\hspace*{-7pt}\end{array}$};
	\node [block] (channel1) [right=0.6cm of encoder1] {$\begin{array}{c}\text{Channel 1}\\\hspace*{-7pt}Q_t^1(Y_t^1|M_{t-1}^1)\hspace*{-7pt}\end{array}$};
	\node [block] (channel2) [right=0.6cm of encoder2] {$\begin{array}{c}\text{Channel 2}\\\hspace*{-7pt}Q_t^2(Y_t^2|M_{t-1}^2)\hspace*{-7pt}\end{array}$};
	\node [block] (receiver) [below right=-0.7 and 0.6cm of channel1] {$\begin{array}{c} \\ \text{\hspace*{-6pt}Receiver\hspace*{-6pt}}\\ \\\hspace*{-6pt}g^3_t\hspace*{-1pt}(Y_{1\hspace*{-1pt}:t-1}^{1:2})\hspace*{-6pt}\\ \\\end{array}$};
	
	\draw [->] ([yshift=11em] source.east) -- node[above] {$X_t^1$} (encoder1.west);
	\draw [->] ([yshift=-11em] source.east) -- node[below] {$X_t^2$} (encoder2.west);
	\draw [->] (encoder1.east) -- node[above] {$M_t^1$} (channel1.west);
	\draw [->] (encoder2.east) -- node[below] {$M_t^2$} (channel2.west);
	\draw [->] (channel1.east) -- node[above] {$Y_t^1$} ([yshift=11em] receiver.west);
	\draw [->] (channel2.east) -- node[above] {$Y_t^2$} ([yshift=-11em] receiver.west);
	\node (end) [right=0.5cm of receiver] {};
	\draw [->] (receiver.east) -- node[above] {$\hat{X}_t$} (end.west);
	
	\end{tikzpicture}}
\end{center}

Each encoder encodes, in real-time, its observations into a sequence of discrete symbols and sends it through a memoryless noisy channel characterized by a transition matrix $Q_t^i(\cdot|\cdot)$, $i\hspace*{-2pt}=\hspace*{-2pt}1,2$.
The receiver wants to construct, in real time,  an estimate $\hat{X}_t$ of the state of the Markov source based on the channels' output $Y_{1:t}^1$,$Y_{1:t}^2$. All agents have the same instantaneous utility given by a distortion function $d_t(X_t,\hat{X}_t)$.       


\vspace{3pt}

\textit{4) Optimal remote and local controller \cite{ouyang2016optimal,asghari2016optimal}:} Consider a decentralized control problem for a Markovian plant with two controllers, a local controller (agent $1$) and a remote controller (agent $2$).   


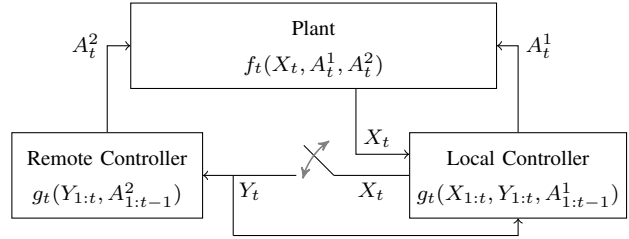
\begin{center}
{\fontsize{8}{9}\selectfont
\begin{tikzpicture}[scale=0.2]

\hspace*{5pt} \node [block] (system) {$\begin{array}{c}\text{\hspace*{-6pt}Plant\hspace*{-6pt}}\\\hspace*{35pt}f_t(X_t,A_t^1,A_t^2)\hspace*{35pt}\end{array}$};
\node [block] (controller1) [below left=0.1cm and 0.9cm of source] {$\begin{array}{c}\text{\hspace*{-6pt}Remote Controller\hspace*{-6pt}}\\g_t(Y_{1:t},A_{1:t-1}^2)\end{array}$};
\node [block] (controller2) [below right=0.1cm and 0.7cm of source] {$\begin{array}{c}\text{\hspace*{-6pt}Local Controller\hspace*{-6pt}}\\\hspace*{-5pt}g_t(X_{1:t},Y_{1:t},A_{1:t-1}^1)\hspace*{-5pt}\end{array}$};

\draw [->] (controller2.north) |- node[name=alabel, right] {$A_t^1$} (system.east);
\draw [->] (controller1.north) |- node[name=alabel, left] {$A_t^2$} (system.west);
\draw [->] ([xshift=10em]system.south) |- node[name=alabel, above right] {$X_t$} ([yshift=5em]controller2.west);
\draw [-] ([yshift=-0em]controller2.west) -- node[name=alabel, below	] {$X_t$} ++ (-5,0cm) coordinate(v1){};
\draw [-] (v1.center) -- ++ (-2,2cm) coordinate(v2){};
\draw [->] (v1.center)+(-2.5,0cm) -- node[name=alabel, below] {$Y_t$} ([yshift=-0em]controller1.east);
\draw [-] (controller1.east)+(2,0cm) -- ++ (2,-4cm) coordinate(v4){};
\draw [->] (v4.center) -| (controller2.south);
\draw[<->,>=stealth',semithick,gray] (v1)+(-0.3,2.3cm) arc (125:165:4.5	cm);
\end{tikzpicture}}
\end{center}

The local controller observes perfectly the state $X_t$ of the Markov chain, and sends his observation through a packet-drop channel to the remote controller. The transmission is successful, \red{\textit{i.e.} $Y_t\hspace*{-2pt}=\hspace*{-2pt}X_t$,} with probability $p\hspace*{-2pt}>\hspace*{-2pt}0$ and is not successful, \red{\textit{i.e.} $Y_t\hspace*{-2pt}=\hspace*{-2pt}\emptyset$,} with probability $1\hspace*{-2pt}-\hspace*{-2pt}p\hspace*{-2pt}\geq\hspace*{-2pt} 0$. We assume that the local controller receives an acknowledgment every time the transmission is successful.  The controllers' joint instantaneous utility is given by a $u_t^{\text{team}}(X_t,\hspace*{-1pt}A_t^1,\hspace*{-1pt}A_t^2)$.     

		\section{Strategies and Beliefs}
\label{sec:equilibrium}

In a dynamic decision problem with asymmetric information 
agents have private information about the evolution of the system, and they do not observe the complete history $\{H_t,X_{t}\}$, $t\hspace*{-2pt}\in\hspace*{-2pt}\mathcal{T}$. Therefore, at every time $t\hspace*{-2pt}\in\hspace*{-2pt}\mathcal{T}$, each agent, say agent $i\hspace*{-2pt}\in\hspace*{-2pt}\mathcal{N}$, needs to form (i) an appraisal about the current state of the system $X_{t}$ and the other agents' information $H_t^{-i}$ (appraisal about the history), and (ii) an appraisal about how other agents will play in the future (appraisal about the future), so as to evaluate the performance of his strategy choices.

When agents are non-strategic, the agents' strategies $g_{1:T}^{1:N}$ are known to all agents. Therefore, agent $i\hspace*{-2pt}\in\hspace*{-2pt}\mathcal{N}$ can form these appraisals by using his private information $H_t^i$ along with \red{the} commonly known strategies $g^{-i}$. Specifically, agent $i$ can utilize his own information $H_t^i$ at $t\hspace*{-2pt}\in\hspace*{-2pt}\mathcal{T}$,  along with (i) the past strategies $g_{1:t-1}$  and (ii) the future strategies $g_{t:T}$ to form these appraisals about the history and the future of the overall system, respectively.
As a result, the outcome of decision problems with non-strategic agents can be fully characterized by the agents' strategy profile $g$.\footnote{We discuss the decision problems with strategic agents in the companion paper \cite{game}.When agents are strategic each agent may have incentive to deviate \textit{an any time} from the strategy the other agents commonly believe he uses if it is profitable to him (see \cite{game} for more discussion).}

However, we need to know the entire strategy profile $g$ for all agents and at all times to form these appraisals so as to evaluate the performance of an arbitrary strategy $g_t^i$, at any time $t\hspace*{-2pt}\in\hspace*{-2pt}\mathcal{T}$ and for any agent $i\hspace*{-2pt}\in\hspace*{-2pt}\mathcal{N}$. Therefore, we must work with the strategy profile $g$ as a whole irrespective of the length of the time horizon $T$. Consequently, the computational complexity of determining a strategy profile that satisfies certain conditions (\textit{e.g.} an optimal strategy profile in teams) grows doubly exponentially in $|\mathcal{T}|$ since the domain of agents' strategy (\textit{i.e.} $|\mathcal{H}_t^i|$) and the number of temporally interdependent decision problems (one for each time instance) grows with $|\mathcal{T}|$. As a result, the analysis of such decision problems is very challenging in general \cite{bernstein2002complexity}.

An alternative conceptual approach for the analysis of decision problems is to define a belief system $\mu$ along with the strategy profile $g$. For every agent $i\hspace*{-2pt}\in\hspace*{-2pt}\mathcal{N}$, at every time $t\hspace*{-2pt}\in\hspace*{-2pt}\mathcal{T}$, define $\mu^i_t(\hspace*{-1pt}h_t^i)$ as the agent $i$'s belief about $\{\hspace*{-1pt}X_t,\hspace*{-1pt}P_t^{-i}\}$ conditioned on the realization of $h_t^i$, that is, $\mu(\hspace*{-1pt}h_t^i)(\hspace*{-1pt}x_t,\hspace*{-1pt}p^{-i})\hspace*{-2pt}:=\hspace*{-2pt}\mathbb{P}^{g_{1:t-1}}\hspace*{-1pt}\{\hspace*{-1pt}X_t\hspace*{-2pt}=\hspace*{-2pt}x_t,\hspace*{-1pt}P_t^{-i}\hspace*{-2pt}=\hspace*{-2pt}p_t^{-i}|h_t^i\}$. The belief $\mu_t^i$ provides an intermediate instrument that encapsulates agent $i$'s appraisal about the past. 
Therefore, agent $i$ can evaluate the performance of any action $a_t^i$ using only the belief $\mu_t^i(h_t^i)$ along with the future strategy profile $g_{t:T}$. However, the belief $\mu(h_t^i)(x_t,\hspace*{-1pt}p^{-i})$ is dependent on  $g_{1:t-1}$ in general since the probability distribution $\mathbb{P}^{g_{1:t-1}}\hspace*{-1pt}\{\hspace*{-1pt}X_t\hspace*{-2pt}=\hspace*{-2pt}x_t,\hspace*{-1pt}P_t^{-i}\hspace*{-2pt}=\hspace*{-2pt}p_t^{-i}|h_t^i\}$ depends on $g_{1:t-1}$. Therefore, the introduction of a belief system offers an equivalent problem formulation that does not necessarily break the inter-temporal dependence between $g_{1:t-1}$ and $g_{t:T}$ and does not simplify the analysis of decision problems. 

Nevertheless, the definition of a belief system has been shown to be suitable for the analysis of single-agent decision making problems (POMDP) for the following reasons.   First, in POMDPs, under perfect recall, the probability distribution $\mathbb{P}^{g_{1:t-1}}\{X_t\hspace*{-2pt}=\hspace*{-2pt}x_t|h_t\}$ is independent of $g_{1:t-1}$; this is \red{known} as \textit{the policy-independence property of beliefs} in stochastic control. Second, the complexity of the belief function does not grow over time since at every \red{time} $t$ the agent only needs to form a belief about $X_t$, which has a time-invariant domain. As a result, we can sequentially decompose  the problem over time to a sequence of static decision problems with time-invariant complexity; such a decomposition leads to a dynamic program. At each stage $t\in\mathcal{T}$ of the dynamic program, we specify $g_t$ by determining an action for each realization of the belief $\mu_t(\cdot)$ fixing the future strategies $g_{t+1:T}$. 
Therefore, the computational complexity of the analysis is reduced from being exponential in $T$ to linear in $T$.

Unfortunately, the above approach for POMDPs does not generalize to decision problems with many agents. This is because of three reasons. First, with many agents, currently in the literature, there exists no information state for each agent that provides a compression of the agent's information, in a mutually consistent manner among the agents, that is sufficient for decision making purposes. Therefore,  an agent's, say agent $i$'s, strategy $g_t^i$ has a growing domain over time. Second, at every time $t\hspace*{-2pt}\in\hspace*{-2pt}\mathcal{T}$, each agent $i\hspace*{-2pt}\in\hspace*{-2pt}\mathcal{N}$ needs to form a belief about the system state $X_t$ as well as the other agents' private information $P_t^{-i}$ that has a growing domain. Therefore the complexity of belief functions grows over time. Third, in decision problems with many agents, the policy-independence \red{property} of belief does not hold in general and  the agents' beliefs at every time $t$ depend on the \red{past strategy profile} $g_{1:t-1}$. Therefore, the agents' beliefs $\mu_t^{1:N}(\cdot)$ are correlated with one another. This correlation depends on $g_{1:t-1}$, and thus, it is not known a priori. Consequently, if we follow an approach similar to that of POMDP to sequentially decompose the problem, we need to solve the decision problem at every stage for every arbitrary correlation among the agents' belief functions, and such a problem is not tractable.\footnote{Alternatively, one can consider arbitrary correlation among the agents' information rather than their beliefs. This is the main idea that underlies the designer's approach proposed by Winstenhausen \cite{witsenhausen1973standard}. Please see Section \ref{sec:disc:c} for more discussion.}  Hence, the methodology proposed for the study of POMDPs is not directly applicable to decision problems with many agents and non-classical information structures.

In the sequel, we propose a notion of sufficient private information and sufficient common information as a mutually consistent compression of the agents' information for decision making purposes. Therefore, we address (partially) the first two problems on the growing domain of the agents' beliefs and strategies. We provide instances of decision problems where we can discover time-invariant information state for each agent. We then utilize the agents' sufficient common information as a coordination instrument, and thus, capture the implicit correlation among the agents' beliefs over time. Accordingly, we present a sequential decomposition of the original decision problems such that at every stage the complexity of the decision problem is similar to that of a static decision multi-agent problem and the size of state variable at every stage is proportional to the dimension of the sufficient private information; thus, we (partially) address the third problem discussed above.

		\vspace*{-5pt}

\section{Sufficient Information}\label{sec:sufficient}
\red{We present the sufficient information approach and characterize an information state that results from compressing the agents' private and common information in a mutually consistent manner. Therefore, we introduce a class of  strategy choices that are simpler than general strategies as they require agents to keep track of only a compressed version of their information over time. We proceed as follows.}
In Section \ref{subsec:privatecompress} we provide  conditions sufficient to determine the subset of private information an agent needs to keep track of over time for decision making purposes.
In Section \ref{subsec:commoncompress}, we introduce the notion of sufficient common information as a compressed version of the agents' common information that along with sufficient private information provides an information state for each agent. We then show, \red{in Section \ref{sec:nonstrategic}}, that this compression of the agents' private and common information provides a sufficient statistic in dynamic decision problems with non-strategic agents. \red{In Section \ref{sec:disc:b}, we provide a generalization of sufficient information approach presented here.}

\vspace*{-5pt}

\subsection{Sufficient Private Information}\label{subsec:privatecompress}
The key ideas for compressing an agent's private information appear in Definitions \ref{def:payoff-relevant} and \ref{def:sufficient} below. To motivate these definitions we first consider the decision problem with single agent, that is, a Partially Observed Markov Decision Process (POMDP), which is a special case of the model described in Section \ref{sec:model} where $N=1$, $H_t^1=P_t^1$ and $C_t=\emptyset$ for all $t\in\mathcal{T}$.

In a POMDP, the agent's belief about the system state $X_t$ conditioned on his history realization $h_t^i$ is an information state. We highlight the three main proprieties that underlie the definition of information state in POMDP (see \cite{mahajan2016decentralized, yuksel2013stochastic}): (1) the information state can be updated recursively, that is, at any time $t$ the information state at $t$ can be written as a function of the information state at $t-1$ and the new information that becomes available at $t$, (2)  the agent's belief about the information state at the next time conditioned on the current information state and action is independent of his information history, and (3) at any time $t$ and for any arbitrary action the agent's expected instantaneous utility conditioned on the information state is independent of his information history.

We generalize the key properties of information state for POMDPs, described above, to decision problems with many agents. We propose a set of conditions sufficient to compress the agents' private information in two steps. First, we consider a decision problem with many agents where there is no signaling among them. Motivated by the definition of information state in POMDPs, we describe conditions sufficient  to determine a compression of the agents' private information (Definition \ref{def:payoff-relevant}). Next, we build on Definition \ref{def:payoff-relevant} \red{as an intermediate conceptual step}, and consider the case where agents are aware of possible signaling among them. Accordingly, we present a set of conditions sufficient  to determine a compression of the agents' private information in decision problems with many agents (Definition \ref{def:sufficient}) .

\red{Therefore,} we first characterize subsets of an agent's private information that are sufficient for the agent's decision making process when there is no signaling among the agents.

\begin{definition}[Private payoff-relevant information]\label{def:payoff-relevant}
	Let $P^{i,pr}_t=\bar{\zeta}_t^i(P_t^i,C_t)$ denote a private signal that agent $i\in\mathcal{N}$ forms at $t\in\mathcal{T}$ based on his private information $P_t^i$ and common information $C_t$. We say $P_t^{i,pr}$ is a \textit{private payoff-relevant information for agent $i$} if, for all open-loop \red{strategy profile} $(A_{1:T}^{1:N}=\hat{a}_{1:T}^{1:N})$ and for all $t\in\mathcal{T}$,
	\begin{enumerate} [(i)]
		\item it can be updated recursively as 
		\begin{gather*}
		P_t^{i,pr}=\bar{\phi}_t^i(P_{t-1}^{i,pr},H_t^i\backslash H_{t-1}^i)\quad \text{if }t\hspace*{-2pt}\neq\hspace*{-2pt}1,
		\end{gather*}
		\item for all realizations $\{c_{t},p_{t}^i\}$ it satisfies
		\begin{gather}
		\hspace*{-18pt}\mathbb{P}^{(\hspace*{-1pt}A_{1:T}^{1:N}\hspace*{-1pt}=\hat{a}_{1:T}^{1:N}\hspace*{-1pt})}\hspace*{-3pt}\left\{\hspace*{-1pt}p_{t+1}^{i,pr}\hspace*{-1pt}\Big|p_t^{i},\hspace*{-1pt}c_{t},\hspace*{-1pt}a_t\hspace*{-2pt}\right\}\hspace*{-3pt}=\hspace*{-2pt}\mathbb{P}^{(\hspace*{-1pt}A_{1:T}^{1:N}\hspace*{-1pt}=\hat{a}_{1:T}^{1:N}\hspace*{-1pt})}\hspace*{-3pt}\left\{\hspace*{-1pt}p_{t+1}^{i,pr}\hspace*{-1pt}\Big|p_t^{i,pr}\hspace*{-1pt},\hspace*{-1pt}c_{t},\hspace*{-1pt}a_t\hspace*{-1pt}\right\}\hspace*{-2pt},\nonumber  
		\end{gather}
		\item for all realizations $\{c_{t},p_{t}^i\}\in\mathcal{C}_{t}\times\mathcal{P}_{t}^i$ such that $\mathbb{P}^{(A_{1:T}^{1:N}=\hat{a}_{1:T}^{1:N})}\{c_t,p_t^i\}>0$,
		\begin{gather} 		
			 \mathbb{E}^{(\hspace*{-1pt}A_{1:t-1}^{1:N}\hspace*{-1pt}=\hat{a}_{1:t}^{1:N}\hspace*{-1pt})}\hspace*{-3pt}\left\{\hspace*{-2pt}	u_t^i\hspace*{-1pt}(\hspace*{-1pt}X_t,\hspace*{-1pt}A_t\hspace*{-1pt})\hspace*{-1pt}\Big|c_t,\hspace*{-1pt}p_t^i,\hspace*{-1pt}a_t\hspace*{-2pt}\right\}\hspace*{-3pt}\nonumber=\hspace*{-1pt}\mathbb{E}^{(\hspace*{-1pt}A_{1:t-1}^{-i}\hspace*{-1pt}=\hat{a}_{1:t}^{-i})}\hspace*{-3pt}\left\{\hspace*{-2pt}u_t^i\hspace*{-1pt}(\hspace*{-1pt}X_t,\hspace*{-1pt}A_t\hspace*{-1pt})\hspace*{-1pt}\Big|c_t,p_t^{i,pr},a_t\right\}\hspace*{-1pt}.\hspace*{-3pt}\nonumber 
		\end{gather} 

	\end{enumerate}
\end{definition}

By assuming that all other agents play open-loop strategies we remove the interdependence between agents $-i$'s strategy choices and agent $i$'s information structure, thus, we eliminate \textit{signaling} among the agents. Fixing the open-loop strategies of agents \hspace*{-1pt}$-i$, agent $i$ faces a centralized stochastic control problem. Definition \ref{def:payoff-relevant} says that $P_t^{i,pr}\hspace*{-1pt}$, $t\hspace*{-2pt}\in\hspace*{-2pt}\mathcal{T}$, is a private payoff-relevant information for agent $i$ if (i) it can be recursively updated, (ii) $P_t^{i,pr}$ includes all information in $P_t^i$ that is relevant to $P_{t+1}^{i,pr}$ and (iii) agent $i$'s instantaneous conditional expected utility at any $t\in\mathcal{T}$ is only a function of $C_t,\hspace*{-1pt}P_t^{i,pr}\hspace*{-1pt}$, and his action $A_t^i$ at $t$. These three conditions are similar to properties (1)-(3) for an information state in POMDP, but they concern only agent $i$'s private information $P_t^i$ instead of the collection  $H_t^i\hspace*{-2pt}=\hspace*{-2pt}\{C_t,\hspace*{-1pt}P_t^i\}$ of his private and common information.\footnote{We note that we interpret a centralized control problem as a special case of our model where $N\hspace*{-2pt}=\hspace*{-2pt}1$, $H_t^1\hspace*{-2pt}=\hspace*{-2pt}P_t$ and $C_t\hspace*{-2pt}=\hspace*{-2pt}\emptyset$ for all $t\hspace*{-2pt}\in\hspace*{-2pt}\mathcal{T}$, Definition \ref{def:payoff-relevant} coincides with the definition of information state for the single agent decision problem.    
We would like to point out that conditions (i)-(iii) can have many solutions including the trivial solution $P_t^{i,pr}\hspace*{-2pt}=\hspace*{-2pt}P_t^i$. \footnote{An interesting research direction is to determine whether a minimal private payoff-relevant information exists, and if so, characterize such a minimal payoff-relevant information. However, such a direction is beyond the scope of this chapter, and we leave this topic for future research.}}

While the definition of private payoff-relevant information suggests a possible way to compress the information required for an agent's decision making process, it assumes that other agents play open-loop strategies  and do not utilize the information they acquire in real-time for decision making purposes (\textit{i.e.} no signaling). However, open-loop strategies are not in general optimal for agents $-i$. As a result, to evaluate the performance of any strategy choice $g^i$ agent $i$ needs also to form a belief about the information that other agents utilize to make decisions.

\begin{definition}[Sufficient private information]
	We say $S_t^i=\zeta_t^i(P_t^i,C_t;g_{1:t-1})$, $i\in\mathcal{N}$, $t\in\mathcal{T}$, is \textit{sufficient private information} for the agents if, 
	\begin{enumerate}[(i)]
		\item it can be updated recursively as \vspace*{-2pt}
		\begin{gather}
		S_t^{i}=\phi_t^i(S_{t-1}^{i},H_t^i\backslash H_{t-1}^i;g_{1:t-1})  \text{ for } t\in\mathcal{T}\backslash\{1\}, \label{eq:sufficientupdate}
		\end{gather}
		\item for any strategy profile $g$ and for all realizations $\{c_t,p_t,p_{t+1},z_{t+1},a_t\}\in\mathcal{C}_t\times\mathcal{P}_t\times\mathcal{P}_{t+1}\times\mathcal{Z}_{t+1}$ of positive probability,\vspace*{-2pt}
		\begin{align}
		\hspace*{-26pt}\mathbb{P}^{g_{1:t}}\hspace*{-1pt}\left\{\hspace*{-2pt}s_{t+1}\hspace*{-1pt},\hspace*{-1pt}z_{t+1}\hspace*{-1pt}\Big|p_t\hspace*{-1pt},\hspace*{-1pt}c_t\hspace*{-1pt},\hspace*{-1pt}a_t\hspace*{-2pt}\right\}\hspace*{-3pt}=\hspace*{-2pt}\mathbb{P}^{g_{1:t}}\hspace*{-1pt}\left\{\hspace*{-2pt}s_{t+1}\hspace*{-1pt},\hspace*{-1pt}z_{t+1}\hspace*{-1pt}\Big|s_t\hspace*{-1pt},\hspace*{-1pt}c_t\hspace*{-1pt},\hspace*{-1pt}a_t\hspace*{-2pt}\right\}\hspace*{-1pt},\hspace*{-4pt}\label{eq:sufficientdynamic}
		\end{align}
		\hspace*{-4pt}where $s_{\tau}^{1:N}\hspace*{-3pt}=\hspace*{-2pt}\zeta_{\tau}^{1:N}\hspace*{-2pt}(p_{\tau}^{1:N}\hspace*{-1pt},\hspace*{-1pt}c_{\tau};\hspace*{-1pt}g_{1\hspace*{-1pt}:\tau-1}\hspace*{-1pt})$ for $\tau\in\mathcal{T}$;
		\item 
		for every strategy profile \red{$\tilde{g}$ of the form} $\tilde{g}\hspace*{-2pt}:=\hspace*{-2pt}\{\hspace*{-1pt}\tilde{g}^{i}_t\hspace*{-1pt}:\hspace*{-1pt}\mathcal{S}_t^i\times \mathcal{C}_t\rightarrow \Delta(\mathcal{A}_t^i), i\hspace*{-2pt}\in\hspace*{-2pt}\mathcal{N}\hspace*{-1pt},\hspace*{-1pt} t\hspace*{-2pt}\in\hspace*{-2pt}\mathcal{T}\}$ and $a_t\hspace*{-2pt}\in\hspace*{-2pt}\mathcal{A}_t$, $t\hspace*{-2pt}\in\hspace*{-2pt}\mathcal{T}$;
		\begin{align} 		
		\hspace*{-26pt}\mathbb{E}^{\tilde{g}_{1:t-1}^{}}\hspace*{-2pt}\left\{\hspace*{-2pt}u_t^i(\hspace*{-1pt}X_t\hspace*{-1pt},\hspace*{-1pt}A_t\hspace*{-1pt})\hspace*{-1pt}\Big|c_t\hspace*{-1pt},\hspace*{-1pt}p_t^i\hspace*{-1pt},\hspace*{-1pt}a_t\hspace*{-2pt}\right\}\hspace*{-3pt}=\hspace*{-2pt}\mathbb{E}^{\tilde{g}_{1:t-1}^{}}\hspace*{-2pt}\left\{\hspace*{-2pt}u_t^i(\hspace*{-1pt}X_t\hspace*{-1pt},\hspace*{-1pt}A_t\hspace*{-1pt})\hspace*{-1pt}\Big|c_t\hspace*{-1pt},\hspace*{-1pt}s_t^{i}\hspace*{-1pt},\hspace*{-1pt}a_t\hspace*{-2pt}\right\}\hspace*{-2pt},\hspace*{-5pt}\label{eq:payoff-relevant2}
		\end{align} 
		for all realizations $\{\hspace*{-1pt}c_{t}\hspace*{-1pt},\hspace*{-1pt}p_{t}^i\}\hspace*{-3pt}\in\hspace*{-2pt}\mathcal{C}_{t}\hspace*{-1pt}\times\hspace*{-1pt}\mathcal{P}_{t}^i$ of positive probability where $s_{\tau}^{1:N}\hspace*{-3pt}=\hspace*{-2pt}\zeta_{\tau}^{1:N}\hspace*{-2pt}(p_{\tau}^{1:N}\hspace*{-1pt},\hspace*{-1pt}c_{\tau};\hspace*{-1pt}\tilde{g}_{1\hspace*{-1pt}:\tau-1}\hspace*{-1pt})$ for $\tau\in\mathcal{T}$;\vspace{5pt}
				 
		\item given an arbitrary strategy profile \red{$\tilde{g}$ of the form} $\tilde{g}\hspace*{-1pt}:=\hspace*{-1pt}\{\tilde{g}^{i}_t:\mathcal{S}_t^i\hspace*{-1pt}\times \hspace*{-1pt}\mathcal{C}_t\rightarrow \Delta(\mathcal{A}_t^i), i\hspace*{-2pt}\in\hspace*{-2pt}\mathcal{N}, t\hspace*{-2pt}\in\hspace*{-2pt}\mathcal{T}\}$, $i\hspace*{-2pt}\in\hspace*{-2pt}\mathcal{N}$, and $t\hspace*{-2pt}\in\hspace*{-2pt}\mathcal{T}$,
		\begin{align}
		\hspace*{-25pt}\mathbb{P}^{\tilde{g}^{}_{1:t-1}}\hspace*{-2pt}\left\{\hspace*{-2pt}s_t^{-i}\hspace*{-1pt}\Big|p_t^i\hspace*{-1pt},\hspace*{-1pt}c_t\hspace*{-2pt}\right\}\hspace*{-3pt}=\hspace*{-2pt}\mathbb{P}^{\tilde{g}^{}_{1:t-1}}\hspace*{-2pt}\left\{\hspace*{-1pt}s_t^{-i}\hspace*{-1pt}\Big|s_t^i\hspace*{-1pt},\hspace*{-1pt}c_t\hspace*{-2pt}\right\}\hspace*{-1pt},\hspace*{-4pt}\label{eq:sufficientinfo}
		\end{align}
		for all realizations $\{c_{t}\hspace*{-1pt},\hspace*{-1pt}p_{t}^i\}\hspace*{-2pt}\in\hspace*{-2pt}\mathcal{C}_{t}\hspace*{-2pt}\times\hspace*{-2pt}\mathcal{P}_{t}^i$ of positive probability where $s_{\tau}^{1:N}\hspace*{-3pt}=\hspace*{-2pt}\zeta_{\tau}^{1:N}\hspace*{-2pt}(p_{\tau}^{1:N}\hspace*{-1pt},\hspace*{-1pt}c_{\tau};\hspace*{-1pt}\tilde{g}_{1\hspace*{-1pt}:\tau-1}\hspace*{-1pt})$ for $\tau\in\mathcal{T}$.		
	\end{enumerate}
	\label{def:sufficient}
\end{definition}
There are four key differences  between the definition of sufficient private information and that of private payoff relevant information. First, we allow that the definition \red{and the update rule} of sufficient information $S_t^i$ to depend on the agents' strategies $g_{1:t-1}$. Second, comparing to part (ii) of Definition \ref{def:payoff-relevant}, part (ii) of Definition \ref{def:sufficient} requires that sufficient information $S_t$ includes all information relevant to the realization of $Z_{t+1}$ in addition to the information relevant to  the realization of $S_{t+1}$. As we discuss further in Section \ref{sec:discussion}, this is because when signaling occurs in a multi-agent decision problems agents need to have a consistent view about future commonly observable events.  Third, comparing part (iii) of Definition \ref{def:sufficient} to part (iii) of Definition \ref{def:payoff-relevant}, we note that the probability measures in Definition \ref{def:sufficient} depend on the strategy profile $g$ \red{instead of the ope-loop strategy profile $(\hspace*{-1pt}A_{1:T}^{1:N}\hspace*{-2pt}=\hspace*{-2pt}\hat{a}_{1:T}^{1:N}\hspace*{-1pt})$}. Fourth, in part (iv) of Definition \ref{def:sufficient} there is  an additional condition requiring that agent $i$'s sufficient private information $S_t^i$ must be rich enough so that he can form beliefs about agents $-i$'s sufficient private information $S_t^{-i}$; such a condition is absent in Definition \ref{def:payoff-relevant}.

In general, the notion of sufficient private information $S_t^{1:N}$ is more restrictive than that of private payoff relevant information $P_t^{1:N,pr}$. This is because,  $S_t^{1:N}$, $t\in\mathcal{T}$, needs to satisfy the additional condition (iv), and furthermore, open-loop strategies are a strict subset of closed loop strategies. 
Definition \ref{def:sufficient} provides (sufficient) conditions under which agents can compress their private information in a ``mutually consistent' manner.  We would like to point out that conditions (i)-(iv) of Definition \ref{def:sufficient} can have many solutions including the trivial solution $S_t^{i}=P_t^i$.\footnote {We do not discuss the possibility of finding a minimal set of sufficient private information in this chapter, and leave it for future research as such investigation is beyond the scope of this chapter.}

\vspace*{-6pt}

\subsection{Sufficient Common Information}\label{subsec:commoncompress}\vspace*{-1pt}

Based on the characterization of sufficient private information, we present a statistic (compressed version) of the common information $C_t$ that agents need to keep track of over time for decision making purposes.

Fix a choice of sufficient private information $S_t^{1:N}$, $t\in\mathcal{T}$. Define $\mathcal{S}_t^i$ to be the set of all possible realizations of $S_t^i$, and $\mathcal{S}_t:=\prod_{i=1}^N\mathcal{S}_t^i$.     
Given the agents' strategy profile $g$, let $\gamma_t:\mathcal{C}_t\rightarrow \Delta(\mathcal{X}_t\times\mathcal{S}_t)$ denote a mapping that determines a conditional probability distribution over the system state $X_t$ and all the agents' sufficient private information $S_t$ conditioned on the common information $C_t$ at time $t$ as\vspace*{-2pt}
\begin{align}
\gamma_t(c_t)(x_t,s_t)=\mathbb{P}^{g_{1:t-1}}\{X_t=x_t,S_t=s_t|c_t\}, \label{eq:SIBcommon belief}\vspace*{-2pt}
\end{align}
for all  $c_t\in\mathcal{C}_t, x_t\in\mathcal{X}_t,s_t\in\mathcal{S}_t$.

We call  the collection of mappings $\gamma:=\{\gamma_t,t\in\mathcal{T}\}$ a \textit{sufficient information based belief system} (SIB belief system). 
Note that $\gamma_t$ is only a function of the common information $C_t$, and thus, it is computable by all agents. Let $\Pi_t^\gamma\hspace*{-2pt}:=\hspace*{-2pt}\gamma_t(C_t)$ denote the (random) common information based belief that agents hold under belief system $\gamma$ at $t$. We can interpret $\Pi_t^\gamma$ as the common belief that each agent holds about the system state $X_t$ and all the agents' (including himself) sufficient private information $S_t$ at time $t$. We call the SIB belief $\Pi_t$ a  \textit{sufficient common information} for the agents.   
In the rest of the paper, we write $\Pi_t$ and drop the superscript $\gamma$ whenever such a simplification in notation is clear. Moreover, we use the terms sufficient common information and SIB belief interchangeably.

\vspace*{-11pt}
 
 \subsection{Sufficient Information based Strategy}\vspace*{-2pt}
 The combination of sufficient private information $S_t^{1:N}$ and sufficient common information (the SIB belief) $\Pi_t$ offers a mutually consistent compression of the agents' private and common  information. 
 Consider a class of strategies that are based on the information given by $(\Pi_t,S_t^i)$ for each agent $i\hspace*{-2pt}\in\hspace*{-2pt}\mathcal{N}$ at time $t\in\mathcal{T}$. We call the mapping $\sigma^i_t\hspace*{-2pt}:\hspace*{-2pt}\Delta(\mathcal{X}_t\times\mathcal{S}_t)\hspace*{-1pt}\times \mathcal{S}_t^i\rightarrow\hspace*{-1pt} \Delta(\mathcal{A}_t^i)$ a \textit{Sufficient Information Based (SIB) strategy} for agent $i$ at time $t$. A SIB strategy $\sigma^i_t$ determines a probability distribution for agent $i$'s action $A_t^i$ at time $t$ given his information $(\Pi_t,S_t^i)$. A SIB strategy is a strategy where agents only use the sufficient common information $\Pi_t=\gamma_t(C_t)$ (instead of \red{complete} common information $C_t$), and the sufficient private information $S_t^i=\zeta_t^i(P_t^i,C_t;g_{1:t-1})$ (instead of complete private information $P_t^i$). A collection of SIB strategies $\{\sigma_{1:T}^1,...,\sigma_{1:T}^N\}$ is called a \textit{SIB strategy profile} $\sigma$.  The set of SIB strategies is a subset of general strategies, defined in Section \ref{sec:model}, as we can define,
 \begin{align}
 g^{(\sigma,\gamma),i}_t(h_t^i):=\sigma^i_t(\pi_t^\gamma,s_t^i) \quad \forall t\hspace*{-2pt}\in\hspace*{-2pt}\mathcal{T}  \label{eq:SIB-strategy}
 \end{align} 
 We note that from Definition \ref{def:sufficient} and (\ref{eq:SIBcommon belief}), the realizations $\pi_t$ and $s_t^{1:N}$ at $t$ only depends on $g_{1:t-1}$. Therefore, strategies $g^{(\sigma,\gamma),i}_t$, defined above via (\ref{eq:SIB-strategy}) needs to be determined iteratively as follows; 
 for $t\hspace*{-2pt}=\hspace*{-2pt}1$, $g^{(\sigma,\gamma),i}_1(h_1^i)\hspace*{-2pt}=\hspace*{-2pt}\sigma^i_1(\pi_1^\gamma,\zeta_1^i(P_1^i,C_1))$;  
 for $t\hspace*{-2pt}=\hspace*{-2pt}2$, $g^{(\sigma,\gamma),i}_2(h_2^i)\hspace*{-2pt}=\hspace*{-2pt}\sigma^i_2(\pi_2^\gamma,\zeta_2^i(P_2^i,C_2;g^{(\sigma,\gamma)}_1))$;
 $...$;
for $t\hspace*{-2pt}=\hspace*{-2pt}T$, $g^{(\sigma,\gamma),i}_t(h_t^i)\hspace*{-2pt}=\hspace*{-2pt}\sigma^i_t(\pi_t^\gamma,\zeta_2^i(P_t^i,C_t;g^{(\sigma,\gamma)}_{t-1}))$.
 Therefore, strategy $g^{(\sigma,\gamma),i}_t$ is well-defined for all $t\hspace*{-2pt}\in\hspace*{-2pt}\mathcal{T}$ and $i\hspace*{-2pt}\in\hspace*{-2pt}\mathcal{N}$. 
 
 \vspace*{-12pt}
 
 \subsection{Sufficient Information based Update Rule}\vspace*{-2pt}
 When the agents play a SIB strategy profile $\sigma$, it is possible to determine the SIB belief $\Pi_t$ recursively over time based on $\Pi_{t-1}$ and the new common information $Z_{t}$ via Bayes' rule. Let $\psi_{t}^{\sigma_{t-1}}\hspace*{-2pt}:\hspace*{-2pt}\Delta(\mathcal{X}_{t-1}\hspace*{-2pt}\times\hspace*{-2pt} \mathcal{S}_{t-1})\hspace*{-2pt}\times\hspace*{-2pt} \mathcal{Z}_{t}\rightarrow \Delta(\mathcal{X}_{t}\hspace*{-2pt}\times\hspace*{-2pt} \mathcal{S}_{t})$ describe such a update rule for time $t+1\in\mathcal{T}/\{1\}$ so that 
 \begin{align}
 \Pi_{t}=\psi_{t}^{\sigma_{t-1}}(\Pi_{t-1},Z_{t}).\label{eq:CIBupdaterule}
 \end{align}

We note that the SIB update rule $\psi_t^{\sigma_{t-1}}$ depends on the SIB strategy profile ${\sigma_{t-1}}$ at $t\hspace*{-2pt}-\hspace*{-2pt}1$. In the rest of the paper, we drop the superscript $\sigma$ whenever such a simplification in notation is clear.

\vspace*{-7pt}
\subsection{Special Cases}
We consider the special cases (1)-(3) of the general model we presented in Section \ref{sec:model}, and identify the sufficient private information  $S_{1:T}^{1:N}$; we discuss the application of sufficient information approach to special case (4) in Section \ref{sec:disc:b}.   

\vspace*{5pt}

\textit{1) Real-time source coding-decoding:} \red{The encoder's and decoders' private information are given by $P_t^1\hspace*{-2pt}=\hspace*{-2pt}\{\hspace*{-1pt}X_{1:t}\hspace*{-1pt}\}$ and  $P_t^2\hspace*{-2pt}=\hspace*{-2pt}\{\hspace*{-1pt}\hat{X}_{1:t-1-\delta}\hspace*{-1pt}\}$, respectively}. The agents' common information is given by $C_t\hspace*{-2pt}=\hspace*{-2pt}\{\hspace*{-1pt}M_{1:t-1}\hspace*{-1pt}\}$. We can verify that $S_t^1=\tilde{X}\hspace*{-2pt}=\hspace*{-2pt}\{\hspace*{-1pt}X_{t-\max(k,\delta+1)+1}\hspace*{-1pt},\hspace*{-1pt}...,\hspace*{-1pt}X_t\hspace*{-1pt}\}$ and $S_t^2=\emptyset$ satisfy the conditions of Definition \ref{def:sufficient} ; this is similar to the structural results in \cite[Sections III and VI]{witsenhausen1979structure}. Consequently, the common information based belief is $\Pi_t\hspace*{-2pt}=\hspace*{-2pt}\mathbb{P}^g\{\hspace*{-1pt}X_{t-\max(k,\delta+1)+1:t}|M_{1:{t-1}}\hspace*{-1pt}\}$.

\vspace*{5pt}

\textit{2) Delayed sharing information structure:} We have $P_t^{i}\hspace*{-2pt}=\hspace*{-2pt}\{\hspace*{-1pt}Y_{t-d+1:t}^i,\hspace*{-1pt}A_{t-d+1:t}^i\hspace*{-1pt}\}$ and $C_t\hspace*{-2pt}=\hspace*{-2pt}\{Y_{1:t-d},$ $A_{1:t-d}\}$. Since we do not assume any specific structure for \red{the} system dynamics and \red{the} agents' observations, agent $i$'s complete private information $P_t^i$ is payoff-relevant for him. Therefore, we set $S_t^i\hspace*{-2pt}=\hspace*{-2pt}P_t^i$. Consequently, we have $\Pi_t\hspace*{-2pt}=\hspace*{-2pt}\mathbb{P}^g\hspace*{-1pt}\{\hspace*{-1pt}X_t,\hspace*{-2pt}Y_{t-d+1:t},\hspace*{-2pt}A_{t-d+1:t}|Y_{1:t-d},\hspace*{-1pt}A_{1:t-d}\hspace*{-1pt}\}$. The above sufficient information appears in the first structural result in \cite{nayyar2011optimal}.

\vspace*{5pt}
\textit{3) Real-time multi-terminal communication:} We have $P_t^1\hspace*{-2pt}=\hspace*{-2pt}\{\hspace*{-1pt}X_{1:t}^1\hspace*{-1pt},\hspace*{-2pt}M_{1:t-1}^1\hspace*{-1pt}\}$,\hspace*{-2pt} $P_t^2\hspace*{-2pt}=\hspace*{-2pt}\{\hspace*{-1pt}X_{1:t}^2\hspace*{-1pt},\hspace*{-2pt}M_{1:t-1}^2\hspace*{-1pt}\}$,  $P_t^3\hspace*{-2pt}=\hspace*{-2pt}\{\hspace*{-1pt}Y_{1:t}^1\hspace*{-1pt},\hspace*{-3pt}Y_{1:t}^2\hspace*{-1pt},\hspace*{-2pt}\hat{X}_{1:t-1}\hspace*{-1pt}\}$, and $C_t\hspace*{-2pt}=\hspace*{-2pt}\emptyset$. It is easy to verify that $S_t^1\hspace*{-2pt}=\hspace*{-2pt}(\hspace*{-1pt}X_t^1\hspace*{-1pt},\hspace*{-2pt}\mathbb{P}\{\hspace*{-1pt}R|X_{1:t}^1\hspace*{-1pt}\}\hspace*{-1pt},\hspace*{-2pt}\mathbb{P}\{Y_{1:t-1}^1|M_{1:t-1}^1\}\hspace*{-1pt})$, $S_t^2\hspace*{-2pt}=\hspace*{-2pt}(\hspace*{-1pt}X_t^2\hspace*{-1pt},\hspace*{-1pt}\mathbb{P}\{\hspace*{-1pt}R|X_{1:t}^2\hspace*{-1pt}\}\hspace*{-1pt},\hspace*{-2pt}\mathbb{P}\{\hspace*{-1pt}Y_{1:t-1}^2|M_{1:t-1}^2\hspace*{-1pt}\}\hspace*{-1pt})$, and $S_t^3\hspace*{-2pt}=\hspace*{-2pt}P_t^3$; \red{this sufficient information corresponds to the structural results  that appear \cite{nayyar2011structure}.}
\vspace{1pt}

\vspace*{-8pt}
  
\section{Main Results}\label{sec:nonstrategic}

In this section, we present our main results for the analysis of dynamic decision problems with asymmetric information and non-strategic agents \red{using the notion of sufficient information}.
We first provide a generalization of the policy-independence property of beliefs to decision problems with many agents (Theorem \ref{thm:beliefindependence}). Second, we show that the set of SIB strategies are rich enough so that restriction to them is without loss of generality (Theorem \ref{thm:nonstrategic}). That is, given any strategy profile $g$, there exists a SIB strategy profile $\sigma$ such that every agent gets the same flow of utility \red{over time} under $\sigma$ as the one under $g$. Third, we consider dynamic team problems with asymmetric information. We show that using the SIB strategies, we can decompose the problem sequentially over time, \red{formulate a dynamic program}, and determine a globally optimal policy via backward induction (Theorem \ref{thm:team}).

\begin{theorem}[Policy-independence belief property]\label{thm:beliefindependence} \hspace*{2pt}\\
		\indent(i) Consider a general strategy profile $g$.
		If agents $-i$ play according to strategies $g^{-i}$, then for every strategy $g^i$ that agent $i$ plays,
		\begin{align}
		\mathbb{P}^{g}\left\{x_{t},p_t^{-i}\Big|h_t^i\right\}=\mathbb{P}^{g^{-i}}\left\{x_{t},p_t^{-i}\Big|h_t^i\right\}.\label{eq:beliefindependence1}
		\end{align}
		
		(ii) Consider a SIB strategy profile $\sigma$ along with the associated update rule $\psi$.
		If agents $-i$ play according to SIB strategies $\sigma^{-i}$, then for every general strategy ${g}^i$ that agent $i$ plays,
		\begin{align}
		\hspace*{-6pt}\mathbb{P}^{\hspace*{-1pt}\sigma^{-i}\hspace*{-2pt},g^i}_\psi\hspace*{-3pt}\left\{\hspace*{-1pt}x_{t},p_t^{-i}\Big|h_t^i\hspace*{-1pt}\right\}\hspace*{-2pt}=\hspace*{-2pt}\mathbb{P}^{\sigma^{-i}}_\psi\hspace*{-3pt}\left\{\hspace*{-1pt}x_{t},p_t^{-i}\Big|h_t^i\hspace*{-1pt}\right\}\hspace*{-2pt}.\label{eq:beliefindependence2}
		\end{align}
\end{theorem}

\vspace*{2pt}

Theorem \ref{thm:beliefindependence} provides a generalization of the \textit {policy-independence belief property} \red{for} the centralized stochastic control \red{problem} \cite{kumar1986stochastic} to multi-agent decision making problems. Part (i) of Theorem \ref{thm:beliefindependence} states 
that, under perfect recall, agent $i$'s belief is independent of his actual strategy $g^i$. Part (ii) of Theorem \ref{thm:beliefindependence} refers to the case where agents $-i$ play SIB strategies $\sigma^{-i}$ and update their SIB belief according to SIB update rule $\psi$. The update rule $\psi$ is determined based on $(\sigma^{-i},\sigma^i)$ via Bayes' rule, where $\sigma^i$ denotes the SIB strategy that agents $-i$ assume agent $i$ utilizes. Equation (\ref{eq:beliefindependence2}) states that even if agent $i$ unilaterally and privately deviates from his SIB strategy, his belief is independent of his actual strategy $g^i$, and only depends on the other agents's strategy $\sigma^{-i}$ as well as the other agents' assumption about the SIB strategy $\sigma^i$ (or equivalently the SIB update rule $\psi$).\footnote{The results of Theorem \ref{thm:beliefindependence} provides a crucial property for the analysis of decision problems with strategic agents. This is because it ensures that an agent's unilateral deviation does not influence his belief (see the companion paper \cite{game} for more details).}

In POMDPs it is shown that restriction to Markov strategies is without loss of optimality. We provide a generalization of this result to decision problems with many agents. 
We show that restriction to SIB strategies is \textit{without loss of generality} \red{in non-strategic settings} given that the agents have access to a \textit{public randomization device}. We say that the agents have access to a public randomization device if at every time $t\hspace*{-2pt}\in\hspace*{-2pt}\mathcal{T}$ they observe a public random signal $\omega_t$ that is completely independent of all events and primitive random variables in the decision problem  and is uniformly distributed on $[0\hspace*{-1pt},\hspace*{-1pt}1]$, and is independent across time. As a result, in general, at every $t\hspace*{-2pt}\in\hspace*{-2pt}\mathcal{T}$, all agents can condition their actions on the realization of $\omega_t$ as well as their own information. \red{In other words, a public randomization device enables the agents to play \textit{correlated} randomized strategies.} We denote by $\sigma_t^i(\Pi_t,S_t^i,\omega_t)$ agent $i$'s SIB strategy using the public randomization device for every $i\in\mathcal{N}$ and $t\in\mathcal{T}$.   

\begin{theorem}\label{thm:nonstrategic}
	Assume that the non-strategic agents have access to a public randomization device. Then, for any strategy profile $g$ there exists an equivalent SIB strategy profile $\sigma$ that results in the same expected flow of utility, \textit{i.e.} 
	\begin{gather}
	\mathbb{E}^g\hspace*{-2pt}\left\{\hspace*{-2pt}\sum_{\tau=t}^{T}\hspace*{-1pt}u_\tau^i(g_\tau^{1:N}\hspace*{-1pt}(\hspace*{-1pt}H_\tau^{1:N}\hspace*{-1pt}),\hspace*{-1pt}X_\tau\hspace*{-1pt})\hspace*{-2pt}\right\}
	=\hspace*{-1pt}\mathbb{E}^\sigma\hspace*{-2pt}\left\{\hspace*{-2pt}\sum_{\tau=t}^{T}\hspace*{-1pt}u^{1:N}_\tau\hspace*{-1pt}(\sigma_\tau^{i}\hspace*{-1pt}(\hspace*{-1pt}\Pi_\tau\hspace*{-1pt},\hspace*{-1pt}S_\tau^{1:N},\omega_\tau),\hspace*{-2pt}X_\tau\hspace*{-1pt})\hspace*{-2pt}\right\}\hspace*{-2pt}, \hspace*{-4pt}
	\end{gather}
	for all $i\hspace*{-2pt}\in\hspace*{-2pt}\mathcal{N}$ and $t\hspace*{-2pt}\in\hspace*{-2pt}\mathcal{T}$.
\end{theorem}

We provide an intuitive explanation for the result of Theorem  \ref{thm:nonstrategic} below. For every agent $i\in\mathcal{N}$, his complete information history $H_t^i$ at any time $t\in\mathcal{T}$ consists of two components: (i) one component captures his information about past events that is relevant to the continuation decision problem; and (ii) another component that, given the first component, captures the information about past events that is irrelevant to the continuation decision problem. We show that the combination of sufficient private information $S_t^i$ and sufficient common information $\Pi_t$ contains the first component. Nevertheless, in general, the agents can coordinate their action by incorporating the second component into their decision since their information about the past events is correlated. Let $R_t^i$
denote the part of agent $i$'s information $H_T^i$ that is not captured by $(\Pi_t,S_t^i)$. We show that the set of $\{R_t^1,...,R_t^N\}$ are jointly independent of $\{(\Pi_t,S_t^i),...,(\Pi_t,S_t^N)\}$ (Lemma \ref{lemma:noise} in the Appendix). Therefore, at every time $t\in\mathcal{T}$, we can generate a set of signals $\{\tilde{R}_t^1,...,\tilde{R}_t^N\}$, one for each agent, using the public randomization device $\omega$ so that they are identically distributed as $\{R_t^1,...,R_t^N\}$. Using the signals $\{\tilde{R}_t^1,...,\tilde{R}_t^N\}$ along with the information state $(\Pi_t,S_t^{i})$ for every agent $i\in\mathcal{N}$, we can thus recreate a (simulated) history that is identically distributed to $H_t^i$. This implies that, given a public randomization device $\omega$, it is sufficient for each agent $i\in\mathcal{N}$ to only keep track of $(\Pi_t,S_t^{i})$ instead of his complete history $H_t^i$, and play a SIB strategy $\sigma^i$ to achieve an identical (in distribution) sequence of outcomes per stage as those under the strategy profile $g$.

The result of Theorem \ref{thm:nonstrategic} states that the the class of SIB strategies \red{characterizes} a set of \textit{simpler} strategies where the agents only keep track of a compressed version of their information rather than their entire information history. Moreover, the restriction to the class of SIB strategies is without loss of generality. Thus, along with results appearing in the companion paper \cite{game}, the result of Theorem \ref{thm:nonstrategic} suggests that the sufficient information approach proposed in this paper presents a unified methodology for the study of decision problems with many non-strategic or strategic agents and asymmetric information.

We would like to discuss the implication of Theorem \ref{thm:nonstrategic} for two special instances of our model. First, when $N=1$, there is no need for a public randomization device since the single decision maker does not need to correlate the outcome of his randomized strategy with any other agent. Therefore, the result of Theorem \ref{thm:nonstrategic} states that the restriction to Markov strategies in POMDPs is without loss of generality. Second, when $N>1$ and the agents have identical utilities, \textit{i.e.} dynamic teams, utilizing a public randomization device does not improve the performance. \red{This is because, in dynamic teams a randomized strategy profile is optimal if and only if it is optimal for every realization of the randomization.} Therefore, the restriction to SIB strategies in dynamic teams is without loss of optimality.

Using the  result of Theorem \ref{thm:nonstrategic}, we present \red{below} a sequential decomposition of dynamic teams over time. We formulate a dynamic program that enables us to determine a globally optimal strategy profile via backward induction.

\begin{theorem} \label{thm:team} A SIB strategy profile $\sigma$ is a globally optimal solution to a dynamic team problem with asymmetric information if it  solves the following dynamic program:
	\begin{align}
	&V_{T+1}(\pi_{t+1}):=0,	\quad\quad	\forall \pi_{t+1},\forall i\in\mathcal{N};\label{eq:team-initial}
	\end{align}
	at every $t\in\mathcal{T}$, and for every $\pi_t$,
	\begin{align}
	&\sigma_t^{1:N}\hspace*{-1pt}(\pi_t,\hspace*{-1pt}\cdot)\hspace*{-2pt}\in \hspace*{-8pt}\argmax_{\alpha^{1:N}: \mathcal{S}_t^{1:N}\rightarrow \Delta(\mathcal{A}_t^{1:N})}\hspace*{-7pt}\mathbb{E}_{\pi_t}\hspace*{-2pt}\big\{\hspace*{-1pt}u_t^{\text{team}}\hspace*{-1pt}(\hspace*{-1pt}X_t\hspace*{-1pt},\hspace*{-1pt}\alpha^{1:N}(S_t^{1:N})\hspace*{-1pt})+V_{t+1}\hspace*{-1pt}(\psi_t^{\sigma_{t}}\hspace*{-1pt}(\hspace*{-1pt}\pi_t\hspace*{-1pt},\hspace*{-1pt}\alpha^{\hspace*{-2pt}1:N}\hspace*{-3pt},\hspace*{-1pt}Z_{t+1}\hspace*{-1pt})\hspace*{-1pt})\hspace*{-2pt}\big\}\hspace*{-1pt},\hspace*{-3pt} \label{eq:team-max}\\
	&V_t(\pi_t)\hspace*{-2pt}:=\hspace*{-6pt}\max_{\alpha^{1:N}: \mathcal{S}_t^{1:N}\rightarrow \Delta(\mathcal{A}_t^{1:N})}\hspace*{-6pt}\mathbb{E}_{\pi_t}\hspace*{-2pt}\big\{\hspace*{-1pt}u_t^{\text{team}}\hspace*{-1pt}(\hspace*{-1pt}X_t\hspace*{-1pt},\hspace*{-1pt}\alpha^{1:N}(S_t^{1:N})\hspace*{-1pt})+V_{t+1}\hspace*{-1pt}(\psi_t^{\sigma_{t}}\hspace*{-1pt}(\hspace*{-1pt}\pi_t\hspace*{-1pt},\hspace*{-1pt}\alpha^{1:N}\hspace*{-3pt},\hspace*{-1pt}Z_{t+1}\hspace*{-1pt})\hspace*{-1pt})\hspace*{-2pt}\big\}\hspace*{-1pt}.\label{eq:team-value}
	\end{align}
\end{theorem}

The results of Theorems \ref{thm:nonstrategic} and \ref{thm:team} extend the results of \cite{nayyar2013decentralized,nayyar2011optimal} for the study of dynamic teams in two directions. First, they state that restriction to  the set of SIB strategies is without loss of generality, while the results of \cite{nayyar2013decentralized,nayyar2011optimal} only state that this restriction is without loss of optimality. Second, the definition of Common Information Based strategies, first presented in \cite{nayyar2013decentralized,nayyar2011optimal}, requires the agents to use all of their private information $P_t^i$, $i\in\mathcal{N}$ (or all their private memory that is a predetermined function of their private information if they do not have perfect recall); the result of Theorem \ref{thm:team} holds for SIB strategies where the agents' private information is effectively compressed , thus, it generalizes/extends the definition of CIB strategies proposed in \cite{nayyar2013decentralized,nayyar2011optimal}.

\vspace*{-5pt}

\section{Discussion} \label{sec:discussion}
\subsection{Constructive algorithm}\label{sec:disc:a}
The sufficient information approach described in Sections \ref{sec:sufficient} and \ref{sec:nonstrategic}, presents a generalization of the notion of information state to dynamic multi-agent decision problems with non-classical information structure. Nevertheless, we would like to point out that our approach does not address all the issues present in the study of dynamic multi-agent decision problems. We discuss the main limitation of our approach below.

In POMDPs, an information state with time-invariant domain can be determined by forming the probability distribution over the system state conditioned on the current information. Our approach does not offer an explicit constructive algorithm that determines a mutually-consistent set of information states, one for each agent, with time-invariant domains in dynamic multi-agent decision problems. Specifically, Definition \ref{def:sufficient}  describes only a set of sufficient conditions that one can use to evaluate whether a specific compression of agents' private information is sufficient for decision making purposes; it does not offer a constructive algorithm to determine a compression of the agents' private information that leads to an information state with time-invariant domain.

Given a set of sufficient private information with time-invariant domain for the agents, we achieve, through the formation of SIB beliefs,  a compression of the agents' common information that results in a set of information states with time-invariant domains. In Sections \ref{sec:model} and \ref{sec:sufficient}, we presented instances of multi-agent decision problems where we can discover a set of information states with time-invariant domains. Nonetheless, it is not clear if such a set of  mutually-consistent information states with time-invariant domains exist for every dynamic multi-agent decision problem. Therefore, an interesting, but challenging, future direction would be to identify classes of dynamic decision problems with non-classical information structure where we can guarantee the existence of a set of mutually-consistent information states with time-invariant domains, and prescribe a constructive methodology for their identification. \red{Moreover, we would like to point out that the sufficient information approach presented here provide sufficient conditions that can be used to evaluate an educated-guess one may have for specific multi-agent problems.}

\vspace*{-10pt}

\subsection{Comparison with other Approaches}	\label{sec:disc:c}
The sufficient information approach proposed in this paper shares similarities and also has differences with existing conceptual approaches to the study of dynamic multi-agent decision problems. Below, we briefly discuss these approaches and compare them with the sufficient information approach. 		

\subsubsection{Comparison with Agent-by-Agent Approach} The agent-by-agent approach proceeds as follows: start with an initial guess of a strategy profile $g$ for all agents. At each iteration, select one agent, say agent $i$. and update his strategy to a best response strategy given the strategy $g^{-i}$ of all other agents. Repeat the process until  a fixed point is reached, that is, when no agent can improve performance by unilaterally changing his strategy. 

If the above-described iterative process converges, the  resulting strategy profile determines an \textit{agent-by-agent optimal strategy profile}; however, such an agent-by-agent optimal strategy profile, in general, is not a globally optimal strategy profile \cite{ho1980team}. This is because the multi-agent decision problems are, in general, not convex in the agents' strategies \cite{mahajan_martins_yuksel}. Therefore, the above-described iterative process does not necessarily converge, or it may converge to a locally optimal strategy profile that is not a globally optimal strategy profile. In contrast to agent-by-agent approach, the sufficient information approach determines a globally optimal strategy profile for multi-agent decision problems with non-strategic agents.   

The agent-by-agent approach can be used to discover qualitative properties of optimal strategies. Specifically, we fix the strategies of all agents except one, say agent $i$, to an arbitrary set of strategies $g^{-i}$, and solve for agent $i$'s best response; to determine agent $i$'s best response we need to solve a POMDP, where the system state and system dynamics, in general, depend on $g^{-i}$. If agent $i$'s best response possesses a property that holds for every choice of $g^{-i}$, then a globally optimal strategy for agent $i$ possesses the same property. In contrast to the agent-by-agent approach \red{where one need to solve a POMDP parameterized by $g^{-i}$}, to discover qualitative properties of a globally optimal strategy profile using the sufficient information approach we only need to check the set of conditions appearing in Definition \ref{def:sufficient} (or equivalently a more general Definition \ref{def:sufficient-general} that will appear in Section \ref{sec:disc:b}). 

Moreover, using the sufficient information approach we can discover qualitative properties of optimal strategies that cannot be discovered by the agent-by-agent approach. 
For instance, consider the following example.

\textit{\textbf{Example.} Consider a team problem with two agents and observable actions, where agent $2$'s action $A_t^2$ does not affect the evolution of $X_t$ for all $t$, \red{\textit{i.e.} $X_{t+1}\hspace*{-2pt}=\hspace*{-2pt}f_t(X_t\hspace*{-1pt},\hspace*{-1pt}A_t^1\hspace*{-1pt},\hspace*{-1pt}W_t)$}. Each agent $i$, $i\hspace*{-2pt}=\hspace*{-2pt}1,2$, has an imperfect private observation of state $X_t$ at $t$ given by $Y_t^i\hspace*{-2pt}=\hspace*{-2pt}O_t^i(X_t,W_t^i)$. An arbitrary choice of strategy $g^{i}_t$ for agent $i$ at $t$ depends, in general, on his complete information history given by $\{Y_{1:t}^i,A_{1:t-1}\}$. Therefore, following the agent-by-agent approach, if agent $i$'s strategy depends on $A_{\tau}^2$ for some $\tau$, $1\leq \hspace*{-2pt}\tau\hspace*{-2pt}\leq t-1$, then agent $j$'s, $j\hspace*{-2pt}\neq\hspace*{-2pt}i$, best response also depends on $A_{\tau}^2$. Consequently, the agent-by-agent approach fails to characterize $A_{1:t-1}^2$ as irrelevant information for decision making purposes for agents $1$ and $2$. However, using the sufficient information approach we can simply show that a globally optimal strategy profile depend only on $\mathbb{P}\{X_t|Y_{1:t}^i,\hspace*{-2pt}A_{1:t}^1\}$ for agent $i$. 
}

\subsubsection{Comparison with the Designer's Approach} The designer's approach was originally proposed by Witsenhausen in \cite{witsenhausen1973standard}, and was further investigated in \cite{mahajan2009optimal}. This approach considers the decision problem from the point of view of a designer (she) who knows the system model and the probability distribution of the primitive random variables, and chooses control/decision strategies for all agents; she chooses these strategies without having any observation/knowledge about the realizations of primitive random variables (\textit{i.e.} she chooses these strategies before the system evolution starts). Therefore, the designer effectively solves a centralized panning problem. The designer's approach proceeds by: (i) formulating the centralized planning problem as a multi-stage, \textit{open-loop} stochastic control problem in which the designer's decision at each time is a set of control strategies for all agents; (ii) using the standard techniques in centralized stochastic control to obtain a dynamic programming decomposition of the decision problem. Each step of the resulting dynamic program is a functional optimization problem. 

 The designer's approach breaks the interdependencies between the agents' decision and information over time by transferring all the complexity that arises due to non-classical information structure and signaling to a larger information state which at each time is \red{given by a} probability distribution on $H_t$, the domain of which \red{increases} with time as agents have perfect recall.\footnote{An instance where the domain of the control law is time-invariant is presented in \cite{mahajan2009optimal}.} Therefore, the sequential decomposition resulting from the designer's approach is not, in general, very practical for the study of multi-agent dynamic decision problems with asymmetric information.  
 
In contrast to the designer's approach, the sufficient information approach provides a sequential decomposition of the decision problem over time where at each time $t$ each agent makes decision based on only a compression of his information $H_t^i$. Therefore, it leads to a dynamic program where the state variable at each step of the program is a probability distribution on $S_t$ instead of a probability distribution on $H_t$ in the designer's approach.

\vspace*{-2pt}

\subsubsection{Comparison with the Common Information Approach}
The common information approach, proposed in \cite{nayyar2011optimal,nayyar2013decentralized}, addresses some of the drawbacks of designer's approach by modeling the decision problem as a \textit{closed-loop} centralized planning problem (POMDP) in which a \textit{coordinator} observes perfectly the common information $C_t$ at each time $t$ and, based on this knowledge, chooses a set of \textit{partial control strategies/prescriptions} that determine how each agent takes an action based on his private information at time $t$. The coordinator's information state at time $t$ is his belief on $(X_t,P_t)$ conditioned $C_t$. As shown in \cite{nayyar2013decentralized}, the dynamic programming decomposition achieved by the common information approach is simpler than that achieved by the designer's approach.
In the common information approach the agents' private information remains intact. Therefore, the resulting decomposition is not very practical whenever the agents' private information grows in time (see special cases 1,3 and 4 in Section \ref{sec:model}). Furthermore, the common information approach becomes identical to the designer's approach whenever the agents do not share any common information over time (see special case 3). 

In the sufficient information approach, we provide conditions sufficient to identify mutually-consistent compressions of the agents' private information that are sufficient for decision making purposes and do not result in any loss in system performance. Thus, the sufficient information approach gives rise to a dynamic program that is simpler than the one resulting from the common information approach. As we show in Section \ref{sec:disc:b}, these conditions are the core of sufficient information approach; they are generalized by Definition \ref{def:sufficient-general} to captures a mutually-consistent joint compressions of the agents' private and common information.      
Moreover, in the model of Section \ref{sec:model}, we do not assume that the agents share a common objective. Therefore,  we do not reformulate the original multi-agent decision problem as a centralized planning problem from the coordinator's point of view when signaling occurs. Alternatively, we provide conditions sufficient to identify compression of the agents' information in a mutually-consistent manner on individual level. As a result, our approach is applicable to both strategic and non-strategic settings (see our companion paper \cite{game} for strategic settings).

\vspace*{-5pt}

\section{Generalization}\label{sec:disc:b}
In the sufficient information approach presented in Section \ref{sec:sufficient}, we treat the agents' private information and common information separately. This is because the main challenge in the study of dynamic decision problems with non-strategic agents is due to the presence of the agents' private information. Nevertheless, such a separate treatment of private and common information is not necessary. Using the same rationale that leads to Definition \ref{def:sufficient}, we present below a set of conditions sufficient to characterize a mutually consistent compression of agents' information, without separating private and common components, that is sufficient for decision making purposes. 

\begin{definition}[Sufficient information]
	We say $L_t^i=\tilde{\zeta}_t^i(P_t^i,C_t,g_{1:t-1})\in\mathcal{L}_t^i$, $i\in\mathcal{N}$, $t\in\mathcal{T}$, is \textit{sufficient information} for the agents if, 
	\begin{enumerate}[(i)]
		\item it can be updated recursively as 
		\begin{gather}
		L_t^{i}=\tilde{\phi}_t^i(L_{t-1}^{i},H_t^i\backslash H_{t-1}^i,g_{1:t-1})  \text{ for } t\in\mathcal{T}\backslash\{1\}, \label{eq:sufficientupdate-gen}
		\end{gather}
		\item for any strategy profile $g$ and for all realizations $\{c_t,p_t,p_{t+1},z_{t+1},a_t\}\in\mathcal{C}_t\times\mathcal{P}_t\times\mathcal{P}_{t+1}\times\mathcal{Z}_{t+1}$ with positive probability,
		\begin{align}
		\hspace*{-26pt}\mathbb{P}^{g_{1:t}}\hspace*{-1pt}\left\{\hspace*{-2pt}l_{t+1}\hspace*{-1pt}\Big|p_t\hspace*{-1pt},\hspace*{-1pt}c_t\hspace*{-1pt},\hspace*{-1pt}a_t\hspace*{-2pt}\right\}\hspace*{-3pt}=\hspace*{-2pt}\mathbb{P}^{g_{1:t}}\hspace*{-1pt}\left\{\hspace*{-2pt}l_{t+1}\hspace*{-1pt}\Big|l_t\hspace*{-1pt},\hspace*{-1pt}a_t\hspace*{-2pt}\right\}\hspace*{-1pt},\hspace*{-4pt}\label{eq:sufficientdynamic-gen}
		\end{align}
		\hspace*{-4pt}where $l_{\tau}^{1:N}\hspace*{-3pt}=\hspace*{-2pt}\tilde{\zeta}_{\tau}^{1:N}\hspace*{-2pt}(p_{\tau}^{1:N}\hspace*{-1pt},\hspace*{-1pt}c_{\tau};\hspace*{-1pt}g_{1\hspace*{-1pt}:\tau-1}\hspace*{-1pt})$ for $\tau\in\mathcal{T}$;	\vspace{2pt}	
		\item 
		for every strategy profile \red{$\tilde{g}$ of the form} $\tilde{g}\hspace*{-2pt}:=\hspace*{-2pt}\{\hspace*{-1pt}\tilde{g}^{i}_t\hspace*{-1pt}:\hspace*{-1pt}\mathcal{L}_t^i\rightarrow \Delta(\mathcal{A}_t^i), i\hspace*{-2pt}\in\hspace*{-2pt}\mathcal{N}\hspace*{-1pt},\hspace*{-1pt} t\hspace*{-2pt}\in\hspace*{-2pt}\mathcal{T}\}$ and $a_t\hspace*{-2pt}\in\hspace*{-2pt}\mathcal{A}_t$, $t\hspace*{-2pt}\in\hspace*{-2pt}\mathcal{T}$;
		\begin{align} 		
		\hspace*{-26pt}\mathbb{E}^{\tilde{g}_{1:t-1}^{}}\hspace*{-2pt}\left\{\hspace*{-2pt}u_t^i(\hspace*{-1pt}X_t\hspace*{-1pt},\hspace*{-1pt}A_t\hspace*{-1pt})\hspace*{-1pt}\Big|c_t\hspace*{-1pt},\hspace*{-1pt}p_t^i\hspace*{-1pt},\hspace*{-1pt}a_t\hspace*{-2pt}\right\}\hspace*{-3pt}=\hspace*{-2pt}\mathbb{E}^{\tilde{g}_{1:t-1}^{}}\hspace*{-2pt}\left\{\hspace*{-2pt}u_t^i(\hspace*{-1pt}X_t\hspace*{-1pt},\hspace*{-1pt}A_t\hspace*{-1pt})\hspace*{-1pt}\Big|\hspace*{-1pt}l_t^{i}\hspace*{-1pt},\hspace*{-1pt}a_t\hspace*{-2pt}\right\}\hspace*{-2pt},\hspace*{-5pt}\label{eq:payoff-relevant2-gen}
		\end{align} 
		for all realizations $\{c_{t}\hspace*{-1pt},\hspace*{-1pt}p_{t}^i\}\hspace*{-2pt}\in\hspace*{-2pt}\mathcal{C}_{t}\hspace*{-2pt}\times\hspace*{-2pt}\mathcal{P}_{t}^i$ of positive probability where $l_{\tau}^{1:N}\hspace*{-3pt}=\hspace*{-2pt}\tilde{\zeta}_{\tau}^{1:N}\hspace*{-2pt}(p_{\tau}^{1:N}\hspace*{-1pt},\hspace*{-1pt}c_{\tau};\hspace*{-1pt}\tilde{g}_{1\hspace*{-1pt}:\tau-1}\hspace*{-1pt})$ for $\tau\in\mathcal{T}$;\vspace{1pt}
		\item given an arbitrary strategy profile \red{$\tilde{g}$ of the form} $\tilde{g}\hspace*{-1pt}:=\hspace*{-1pt}\{\tilde{g}^{i}_t:\mathcal{L}_t^i\hspace*{-1pt}\rightarrow \Delta(\mathcal{A}_t^i), i\hspace*{-2pt}\in\hspace*{-2pt}\mathcal{N}, t\hspace*{-2pt}\in\hspace*{-2pt}\mathcal{T}\}$, $i\hspace*{-2pt}\in\hspace*{-2pt}\mathcal{N}$, and $t\hspace*{-2pt}\in\hspace*{-2pt}\mathcal{T}$,
		\begin{align}
		\hspace*{-25pt}\mathbb{P}^{\tilde{g}^{}_{1:t-1}}\hspace*{-2pt}\left\{\hspace*{-2pt}l_t^{-i}\hspace*{-1pt}\Big|p_t^i\hspace*{-1pt},\hspace*{-1pt}c_t\hspace*{-2pt}\right\}\hspace*{-3pt}=\hspace*{-2pt}\mathbb{P}^{\tilde{g}^{}_{1:t-1}}\hspace*{-2pt}\left\{\hspace*{-1pt}l_t^{-i}\hspace*{-1pt}\Big|l_t^i\hspace*{-2pt}\right\}\hspace*{-1pt},\hspace*{-4pt}\label{eq:sufficientinfo-gen}
		\end{align}
		for all realizations $\{c_{t}\hspace*{-1pt},\hspace*{-1pt}p_{t}^i\}\hspace*{-2pt}\in\hspace*{-2pt}\mathcal{C}_{t}\hspace*{-2pt}\times\hspace*{-2pt}\mathcal{P}_{t}^i$  with positive probability where $l_{\tau}^{1:N}\hspace*{-3pt}=\hspace*{-2pt}\tilde{\zeta}_{\tau}^{1:N}\hspace*{-2pt}(p_{\tau}^{1:N}\hspace*{-1pt},\hspace*{-1pt}c_{\tau};\hspace*{-1pt}\tilde{g}_{1\hspace*{-1pt}:\tau-1}\hspace*{-1pt})$ for $\tau\in\mathcal{T}$.		
	\end{enumerate}
	\label{def:sufficient-general}
\end{definition}

The conditions of Definition \ref{def:sufficient-general} are similar to those of Definition \ref{def:sufficient}, but they concern agents' private \textit{and} common information rather than just their private information. Throughout the paper, we do not make any assumption that the agents' private observations are necessarily \red{disjoint}. Therefore, one can define $P_t^i=H_t^i$ and $C_t^i=\emptyset$, for all $i\in\mathcal{N}$ and $t\in\mathcal{T}$, in which case Definition \ref{def:sufficient-general} would be the same as Definition \ref{def:sufficient}. Consequently, all the results appearing in this paper (Theorems \ref{thm:beliefindependence}-\ref{thm:team-inifnite}) also hold for sufficient information characterized by Definition \ref{def:sufficient-general}.  			

We show below that the set of information  states $(S_t^{i},\Pi_t)$, $i\in\mathcal{N}\}$ proposed in Section \ref{sec:sufficient} satisfies the conditions of Definition \ref{def:sufficient-general}. Therefore, Definition \ref{def:sufficient-general} provides a generalization of the sufficient information approach presented in Section \ref{sec:sufficient} as it does not require to compress the agents' private and common information separately. 

\begin{theorem}\label{lemma:general}
	The set of information states $L_t^i:=(S_t^{i},\Pi_t)$, $i\in\mathcal{N}$, $t\in\mathcal{T}$, satisfies Definition \ref{def:sufficient-general}.  
\end{theorem}			
Compared to Definition \ref{def:sufficient}, Definition \ref{def:sufficient-general} provides conditions sufficient for a mutually-consistent joint compression of the agents' private and common information. However, similar to the discussion in Section \ref{sec:disc:a}, it does not provide a constructive algorithm to determine a set of sufficient information $L_t^i$, $i\in\mathcal{N}$, $t\in\mathcal{T}$, with time-invariant domain. 

\begin{remark}
	In view of Definition \ref{def:sufficient-general}, one can replace condition (ii) of Definition \ref{def:sufficient} with a weaker one that requires that $S_t$ include all the information necessary to form a belief about the realizations (of parts) of $Z_{t+1}$ only if (those parts of) $Z_{t+1}$ affect the realization of $\Pi_{t+1}$ given $\Pi_t$.      
\end{remark}

\red{Using Definition \ref{def:sufficient-general} we identify a set of sufficient information for special case 4 described in Section \ref{sec:model}.}

\textbf{Special Case:}

\textit{4) \hspace*{-1pt}Optimal remote and local controller:} We have $C_t\hspace*{-2pt}=\hspace*{-2pt}\{\hspace*{-1pt}Y_{1:t}\hspace*{-1pt}\}$, $P_t^1\hspace*{-2pt}=\hspace*{-2pt}\{\hspace*{-1pt}X_{1:t}\hspace*{-1pt},\hspace*{-2	pt}A_{1:t-1}^1\hspace*{-1pt}\}\hspace*{-1pt}\backslash\hspace*{-1pt} C_t$, and $P_t^2\hspace*{-2pt}=\hspace*{-2pt}\{\hspace*{-1pt}A_{1:t-1}^2\hspace*{-1pt}\}$. Let $\tau\hspace*{-2pt}\leq\hspace*{-2pt} t$ denote the last time the data transmission was successful between the local and remote controllers. \red{We can restrict attention, without loss of optimality,} to the class \red{of} pure strategies for both controllers. Therefore, one can show that $L_t^1\hspace{-2pt}=\hspace*{-2pt}\{\hspace*{-1pt}X_t\hspace*{-1pt},\hspace*{-1pt}\{\mathbb{P}^g\hspace*{-1pt}\{\hspace*{-1pt}X_t\hspace*{-2pt}=\hspace*{-2pt}x_t|X_{\hat{\tau}}\hspace*{-1pt}\}\hspace*{-1pt},\hspace*{-1pt}\forall x_t\hspace*{-2pt}\in\hspace*{-2pt}\mathcal{X}_t\hspace*{-1pt}\}\hspace*{-2pt}\}$ and $L_t^2\hspace*{-2pt}=\hspace*{-2pt}\{\mathbb{P}^g\hspace*{-1pt}\{\hspace*{-1pt}X_t\hspace*{-2pt}=\hspace*{-2pt}x_t|X_{\hat{\tau}}\hspace*{-1pt}\}\hspace*{-1pt},\hspace*{-1pt}\forall x_t\hspace*{-2pt}\in\hspace*{-2pt}\mathcal{X}_t\hspace*{-1pt}\}$ satisfy the conditions of Definition \hspace*{-1pt}\ref{def:sufficient-general}; this is similar to the structural results in \cite{ouyang2016optimal,asghari2016optimal}.

\vspace*{-3pt}

		\section{Extension to Infinite Horizon} \label{sec:infinite}
In the model of Section \ref{sec:model}, we assume that the horizon $T$ is finite. We present a model similar to that of Section \ref{sec:model} with infinite horizon, \textit{i.e.} $T=\infty$, and provide  the extensions of our results to dynamic decision problems with infinite horizon.

\vspace{10pt}

\textbf{Infinite Horizon Dynamic Decision Problem:} There are $N$ non-strategic agents who live in a dynamic Markovian world over an infinite horizon. Consider a time-invariant model where the system state, actions, and observations spaces are finite and time-invariant, \textit{i.e.} $\mathcal{X}_{\infty}=\mathcal{X}_t$, $\mathcal{A}_{\infty}=\mathcal{A}_t$, $\mathcal{Z}_{\infty}=\mathcal{Z}_t$, and $\mathcal{Y}_{\infty}=\mathcal{Y}_t$ for all $t\in\mathbb{N}$. Let $X_t\in\mathcal{X}_{\infty}$ denote the system state at $t\in\mathbb{N}$. Given the agents' actions $A_t$ at $t$, the system state evolution is given by
\begin{align}
X_{t+1}=f_{\infty}(X_t,A_t,W_t^x),
\end{align} 
where $\{W_t^x,t\in\mathbb{N}\}$ is a sequence of independent and identically distributed random variables. The initial state $X_1$ is a random variable with probability distribution $\eta\in\Delta(\mathcal{X}_{\infty})$ with full support that is common knowledge among the agents. 

At every time $t\in\mathbb{N}$, each agent $i\in\mathcal{N}$, receives a noisy observation $Y_t^i$ given by
\begin{align}
Y_{t}^i=O_{\infty}^i(X_t,A_{t-1},W_t^i),
\end{align}
where $\{W_t^i,t\in\mathbb{N},i\in\mathcal{N}\}$ is a sequence of independent and identically distributed random variables.

In addition, at every $t\in\mathbb{N}$ all agents receive a common observation $Z_t\in\mathcal{Z}_\infty$ given by
\begin{align}
Z_{t}=O_{\infty}^c(X_t,A_{t-1},W_t^c),
\end{align}
where $\{W_t^c,t\in\mathbb{N}\}$ is a sequence of independent and identically distributed random variables; the sequences $\{W_t^x,t\in\mathbb{N}\}$, $\{W_t^c,t\in\mathbb{N}\}$, and $\{W_t^i,t\in\mathbb{N},i\in\mathcal{N}\}$ and the initial state $X_1$ are mutually independent. 

Similar to the model of Section \ref{sec:model}, let $P_t^i$ and $C_t$ denote agent $i$'s, $i\in\mathcal{N}$, private and common information at $t\in\mathbb{N}$, respectively. Agent $i$ has a \textit{time-invariant} instantaneous utility function $\delta^{t-1}u^i_{\infty}(X_t,A_t)$, and his total discounted utility is given by
\begin{align}
U^i_{\text{in}}(X,A):=\sum_{t=1}^\infty \delta^{t-1}u^i_{\infty}(X_t,A_t),
\end{align}
where $\delta$ denotes the discount factor. 

\vspace{10pt}

We provide an extension of our results to infinite horizon dynamic decision problems with non-strategic agents. For that matter, we first present a generalization of the definition of sufficient private information to infinite horizon decision problems.

\begin{definition}[Time-invariant sufficient private information] We say $S^i_{t}$, $i\in\mathcal{N}$, $t\in\mathbb{N}$, is a time-invariant sufficient private information if it is a sufficient private information and has a time-invariant domain denoted by $\mathcal{S}_{\infty}^i$, $i\in\mathcal{N}$.  
\label{def:SIB-inf}
\end{definition}

We note that for the special cases presented in Section \ref{sec:sufficient}, the characterized sufficient private information  is time-invariant.

Following an argument similar to the one presented in Section \ref{sec:nonstrategic}, we extend the result of Theorem \ref{thm:nonstrategic} to infinite horizon dynamic decision problems with non-strategic agents.

\begin{theorem}\label{thm:nonstrategic-infinite}
	Consider an infinite horizon dynamic decision problem with non-strategic agents having access to a public randomization device. Then, for any arbitrary strategy profile $g$ there exists an equivalent stationary SIB strategy profile $\sigma$ that results in the same expected flow of utility, \textit{i.e.},
	\begin{gather}
	\mathbb{E}^g\hspace*{-2pt}\left\{\hspace*{-2pt}\sum_{\tau=t}^{\infty}\hspace*{-1pt}\delta^{t-1}u_{\infty}^i(g_\tau^{1:N}\hspace*{-1pt}(\hspace*{-1pt}H_\tau^{1:N}\hspace*{-1pt}),\hspace*{-1pt}X_\tau\hspace*{-1pt})\hspace*{-2pt}\right\}=
	\mathbb{E}^{\sigma_{\infty}}\hspace*{-2pt}\left\{\hspace*{-2pt}\sum_{\tau=t}^{\infty}\hspace*{-1pt}\delta^{t-1}u^{i}_{\infty}\hspace*{-1pt}(\sigma_{\tau}^{1:N}\hspace*{-1pt}(\hspace*{-1pt}\Pi_\tau\hspace*{-1pt},\hspace*{-1pt}S_\tau^{1:N},\omega_\tau),\hspace*{-2pt}X_\tau\hspace*{-1pt})\hspace*{-1pt}\right\} \label{eq:nonstrategic-infinite}, 
	\end{gather}
	for all $i\in\mathcal{N}$ and $t\in\mathbb{N}$. 
\end{theorem}

Next, we consider the case where agents share the same objective $u^i_{\infty}(\cdot,\cdot)=u^{\text{team}}_{\infty}(\cdot,\cdot)$ for all $i\in\mathcal{N}$., \textit{i.e.} \red{an infinite horizon}  dynamic team problem. \red{It is shown that in infinite horizon POMDPS we can restrict attention, without loss of generality, to stationary Markov policies \cite{kumar1986stochastic}. We provide a generalization of this results to dynamic multi-agent decision problems below. }  

Given a set of time-invariant sufficient private information, let  $\Pi_t\hspace*{-2pt}\in\hspace*{-2pt}\Delta(\mathcal{X}_\infty\hspace*{-2pt}\times\hspace*{-2pt}\mathcal{S}_{\infty}) $ denote the SIB belief about $(X_t,S_{t})$ at time $t$. 
We call the mapping $\sigma_{\infty}^i:\Delta(\mathcal{X}_\infty\times\mathcal{S}_{\infty})\times\mathcal{S}_{\infty}^i \rightarrow \Delta(\mathcal{A}_{\infty}^i)$ a \textit{stationary SIB strategy} for agent $i$ if $S_t^i$, $i\in\mathcal{N}$, $t\in\mathbb{N}$, is a time-invariant sufficient private information. 
Similarly, 
given a stationary SIB strategy profile $\sigma_{\infty}$, we define a \textit{stationary SIB update rule} as a time-invariant mapping $\eta_{\infty}^{\sigma_{\infty}}:\Delta(\mathcal{X}_{\infty}\times\mathcal{S}_{\infty})\times \mathcal{Z}_\infty\rightarrow \Delta(\mathcal{X}_{\infty}\times\mathcal{S}_{\infty})$, that recursively determines the SIB belief via Bayes' rule for all $t\in\mathbb{N}$.  Similarly, let $\sigma_\infty^i(\Pi_t,S_t^i,\omega_t)$ denote agent $i$'s \textit{stationary} SIB strategy using the public randomization device for every $i\hspace*{-2pt}\in\hspace*{-2pt}\mathcal{N}$ and $t\hspace*{-2pt}\in\hspace*{-2pt}\mathcal{T}$, when the agents have access to a public randomization device $\omega_t$ for every $t\hspace*{-2pt}\in\hspace*{-2pt}\mathcal{T}$.

We provide a sequential decomposition similar to that of Theorem \ref{thm:team} for infinite horizon dynamic teams below. 

\begin{theorem} \label{thm:team-inifnite} A stationary SIB strategy profile $\sigma_{\infty}$ is an optimal solution to an infinite horizon dynamic team problem with asymmetric information if it solves the following Bellman equation:
		\begin{align}
		&V_{\infty}(\pi_t)\hspace*{-2pt}:=\hspace*{-6pt}\max_{\alpha^{1:N}: \mathcal{S}_{\infty}^{1:N}\rightarrow \mathcal{A}_\infty^{1:N}}\hspace*{-6pt}\mathbb{E}_{\pi}\hspace*{-2pt}\big\{\hspace*{-1pt}u_{\infty}^{\text{team}}\hspace*{-1pt}(\hspace*{-1pt}X_t\hspace*{-1pt},\hspace*{-1pt}\alpha^{1:N}(S_{t}^{1:N})\hspace*{-1pt})+V_{\infty}\hspace*{-1pt}(\eta_{\infty}\hspace*{-1pt}(\hspace*{-1pt}\pi_t\hspace*{-1pt},\hspace*{-1pt}\alpha^{1:N}\hspace*{-3pt},\hspace*{-1pt}Z_{t+1}\hspace*{-1pt})\hspace*{-1pt})\hspace*{-1pt}\big\}\hspace*{-1pt},\label{eq:team-infinite-value}
		\end{align}
		\hspace*{3pt}for all $\pi_t\in\Delta(\mathcal{X}_{\infty}\times \mathcal{S}_{\infty})$.
	\end{theorem}

The result of Theorem \ref{thm:team-inifnite} provide a generalization of Bellman equation for POMDPS  (see \cite[Ch. 8]{kumar1986stochastic}) to decision problems with many agents and asymmetric information.
		\section{Conclusion} \label{sec:conclsion}
		We presented a general approach to study a general class of dynamic multi-agent decision making problems with non-strategic agents. We proposed the notion of sufficient information that enables us to compress effectively the agents' (private and common) information in a mutually consistent manner for decision making purposes. We showed that the restriction to the class of SIB strategies are without loss of generality. Accordingly, we provided a sequential decomposition of dynamic decision problems with non-strategic agents, and formulated a dynamic program to determine a globally optimal strategy profile in dynamic teams. The proposed sufficient information approach presented in this paper generalizes a set of existing results in the literature for the study of dynamic multi-agent decision making problems with non-strategic agents. Our results in this paper, along with those appearing in the companion paper \cite{game}, provide a unified appraoch to study dynamic decision problems with non-strategic agents (teams) and strategic agents (games).
		For future directions, we will investigate the problem of determining a constructive algorithm that enables us to identify sufficient (private) information in a systematic way.
	
		\bibliographystyle{ieeetr}
		\bibliography{thesis-bib,collection,otherrefs}

\begin{thebibliography}{10}

\bibitem{CDC18}
H.~Tavafoghi, Y.~Ouyang, and D.~Teneketzis, ``A sufficient information approach
  to decentralized decision making,'' in {\em 57th IEEE Conference on Decision
  and Control (CDC)}, 2018.

\bibitem{game}
H.~Tavafoghi, Y.~Ouyang, and D.~Teneketzis, ``A unified approach to dynamic
  multi-agent decision problems with asymmetric information - part i: Strategic
  agents,'' {\em working paper}, 2018.

\bibitem{kumar1986stochastic}
P.~Kumar and P.~Varaiya, {\em Stochastic Systems: Estimation Identification and
  Adaptive Control}.
\newblock Prentice-Hall, Inc., 1986.

\bibitem{Bertsekas:1995}
D.~P. Bertsekas, {\em Dynamic Programming and Optimal Control}, vol.~1.
\newblock Belmont, MA: Athena Scientific, 1995.

\bibitem{mahajan_martins_yuksel}
A.~Mahajan, N.~C. Martins, M.~C. Rotkowitz, and S.~Y{\"u}ksel, ``Information
  structures in optimal decentralized control,'' in {\em 51st IEEE Conference
  on Decision and Control (CDC)}, pp.~1291--1306, 2012.

\bibitem{kulkarni2015optimizer}
A.~A. Kulkarni and T.~P. Coleman, ``An optimizer's approach to stochastic
  control problems with nonclassical information structures,'' {\em IEEE
  Transactions on Automatic Control}, vol.~60, no.~4, pp.~937--949, 2015.

\bibitem{lessard2016convexity}
L.~Lessard and S.~Lall, ``Convexity of decentralized controller synthesis,''
  {\em IEEE Transactions on Automatic Control}, vol.~61, no.~10,
  pp.~3122--3127, 2016.

\bibitem{yuksel2016convex}
S.~Y{\"u}ksel and N.~Saldi, ``Convex analysis in decentralized stochastic
  control and strategic measures,'' in {\em 55th IEEE Annual Conference on
  Decision and Control (CDC)}, pp.~6050--6055, 2016.

\bibitem{Witsenhausen:1968}
H.~S. Witsenhausen, ``A counterexample in stochastic optimum control,'' {\em
  {SIAM} Journal of Optimal Control}, vol.~6, no.~1, pp.~131--147, 1968.

\bibitem{Chu:1972}
Y.-C. Ho and K.-C. Chu, ``Team decision theory and information structures in
  optimal control problems--part i,'' {\em IEEE Transactions on Automatic
  Control}, vol.~17, no.~1, pp.~15--22, 1972.

\bibitem{LamperskiDoyle:2011}
A.~Lamperski and J.~C. Doyle, ``On the structure of state-feedback lqg
  controllers for distributed systems with communication delays,'' in {\em 50th
  IEEE Conference on Decision and Control and European Control Conference
  (CDC-ECC)}, pp.~6901--6906, 2011.

\bibitem{LessardNayyar:2013}
L.~Lessard and A.~Nayyar, ``Structural results and explicit solution for
  two-player {LQG} systems on a finite time horizon,'' in {\em 52nd IEEE
  Conference on Decision and Control (CDC)}, pp.~6542--6549, 2013.

\bibitem{ShahParrilo:2013}
P.~Shah and P.~Parrilo, ``{${\cal H}_{2}$}-optimal decentralized control over
  posets: A state-space solution for state-feedback,'' vol.~58, pp.~3084--3096,
  Dec. 2013.

\bibitem{Nayyar_Lessard_2015}
A.~Nayyar and L.~Lessard, ``Structural results for partially nested {LQG}
  systems over graphs,'' in {\em American Control Conference (ACC), 2015},
  pp.~5457--5464, 2015.

\bibitem{Lessard_Lall_2015}
L.~Lessard and S.~Lall, ``Optimal control of two-player systems with output
  feedback,'' {\em IEEE Transactions on Automatic Control}, vol.~60, no.~8,
  pp.~2129--2144, 2015.

\bibitem{Yuksel:2009}
S.~Yuksel, ``Stochastic nestedness and the belief sharing information
  pattern,'' {\em IEEE Transactions on Automatic Control}, vol.~54, no.~12,
  pp.~2773--2786, 2009.

\bibitem{Ouyang_Asghari_Nayyar:2017}
Y.~Ouyang, S.~M. Asghari, and A.~Nayyar, ``Stochastic teams with randomized
  information structures,'' in {\em 56th IEEE Conference on Decision and
  Control (CDC)}, 2017.

\bibitem{nayyar2011optimal}
A.~Nayyar, A.~Mahajan, and D.~Teneketzis, ``Optimal control strategies in
  delayed sharing information structures,'' {\em IEEE Transactions on Automatic
  Control}, vol.~56, no.~7, pp.~1606--1620, 2011.

\bibitem{witsenhausen1971separation}
H.~Witsenhausen, ``Separation of estimation and control for discrete time
  systems,'' {\em Proceedings of the IEEE}, vol.~59, no.~11, pp.~1557--1566,
  1971.

\bibitem{varaiya1978delayed}
P.~Varaiya and J.~Walrand, ``On delayed sharing patterns,'' {\em IEEE
  Transactions on Automatic Control}, vol.~23, no.~3, pp.~443--445, 1978.

\bibitem{yoshikawa1978decomposition}
T.~Yoshikawa, ``Decomposition of dynamic team decision problems,'' {\em IEEE
  Transactions on Automatic Control}, vol.~23, no.~4, pp.~627--632, 1978.

\bibitem{RotkowitzLall:2006}
M.~Rotkowitz and S.~Lall, ``A characterization of convex problems in
  decentralized control,'' {\em IEEE Transactions on Automatic Control},
  vol.~50, no.~12, pp.~1984--1996, 2005.

\bibitem{asghari2016dynamic}
S.~M. Asghari and A.~Nayyar, ``Dynamic teams and decentralized control problems
  with substitutable actions,'' 2016.

\bibitem{ho1980team}
Y.~Ho, ``Team decision theory and information structures,'' {\em Proceedings of
  the IEEE}, vol.~68, no.~6, pp.~644--654, 1980.

\bibitem{witsenhausen1973standard}
H.~S. Witsenhausen, ``A standard form for sequential stochastic control,'' {\em
  Mathematical Systems Theory}, vol.~7, no.~1, pp.~5--11, 1973.

\bibitem{nayyar2013decentralized}
A.~Nayyar, A.~Mahajan, and D.~Teneketzis, ``Decentralized stochastic control
  with partial history sharing: A common information approach,'' {\em IEEE
  Transactions on Automatic Control}, vol.~58, no.~7, pp.~1644--1658, 2013.

\bibitem{witsenhausen1979structure}
H.~Witsenhausen, ``On the structure of real-time source coders,'' {\em The Bell
  System Technical Journal}, vol.~58, no.~6, pp.~1437--1451, 1979.

\bibitem{kurtaran1979corrections}
B.~Kurtaran, ``Corrections and extensions to" decentralized stochastic control
  with delayed sharing information pattern",'' {\em IEEE Transactions on
  Automatic Control}, vol.~24, no.~4, pp.~656--657, 1979.

\bibitem{nayyar2011structure}
A.~Nayyar and D.~Teneketzis, ``On the structure of real-time encoding and
  decoding functions in a multiterminal communication system,'' {\em IEEE
  Transactions on Information Theory}, vol.~57, no.~9, pp.~6196--6214, 2011.

\bibitem{ouyang2016optimal}
Y.~Ouyang, S.~Asghari, and A.~Nayyar, ``Optimal local and remote controllers
  with unreliable communication,'' in {\em 55th IEEE Conference on Decision and
  Control (CDC)}, pp.~6024--6029, 2016.

\bibitem{asghari2016optimal}
S.~M. Asghari, Y.~Ouyang, and A.~Nayyar, ``Optimal local and remote controllers
  with unreliable uplink channels,'' {\em IEEE Transactions on Automatic
  Control}, forthcoming.

\bibitem{bernstein2002complexity}
D.~S. Bernstein, R.~Givan, N.~Immerman, and S.~Zilberstein, ``The complexity of
  decentralized control of markov decision processes,'' {\em Mathematics of
  operations research}, vol.~27, no.~4, pp.~819--840, 2002.

\bibitem{mahajan2016decentralized}
A.~Mahajan and M.~Mannan, ``Decentralized stochastic control,'' {\em Annals of
  Operations Research}, vol.~241, no.~1-2, pp.~109--126, 2016.

\bibitem{yuksel2013stochastic}
S.~Y{\"u}ksel and T.~Ba{\c{s}}ar, {\em Stochastic {N}etworked {C}ontrol
  {S}ystems: {S}tabilization and {O}ptimization under {I}nformation
  {C}onstraints}.
\newblock Springer Science \& Business Media, 2013.

\bibitem{mahajan2009optimal}
A.~Mahajan and D.~Teneketzis, ``Optimal design of sequential real-time
  communication systems,'' {\em IEEE Transactions on Information Theory},
  vol.~55, no.~11, pp.~5317--5338, 2009.

\end{thebibliography}

		\appendix
\vspace*{-5pt}
\begin{proof}[\textbf{Proof of Theorem \ref{thm:beliefindependence}}]	
	We prove the result of part (i) by induction.
	For $t\hspace*{-1pt}=\hspace*{-1pt}1$ the result holds since the agents have not taken any action yet.
	Suppose that (\ref{eq:beliefindependence1}) holds for $t\hspace*{-1pt}-\hspace*{-1pt}1$. Then,
	\begin{align}
	&\mathbb{P}^{g}\hspace{-3pt}\left\{\hspace{-1pt}x_t,\hspace{-1pt}h_t^{-i}|h_t^i\hspace{-1pt}\right\}\hspace{-1pt}=\hspace{-1pt}\sum_{x_{t-1}}\hspace{-1pt}\mathbb{P}^{g}\hspace{-2pt}\left\{\hspace{-1pt}x_t,x_{t-1},h_t^{-i}|h_t^i\hspace{-1pt}\right\}
	=
	\sum_{x_{t-1}}\hspace{-1pt}\mathbb{P}^{g}\left\{\hspace{-1pt}x_t,x_{t-1},h_{t-1}^{-i},a_{t-1}^{-i},y_t^{-i}|h_{t-1}^i,a_{t-1}^i,y_t^i,z_t\hspace{-1pt}\right\}\nonumber\\
	=&
	\hspace{-2pt}\sum_{x_{t-1}}\hspace{-1pt}\mathbb{P}\{\hspace{-1pt}y_t^{-i}|x_{t}\hspace{-1pt},\hspace{-1pt}a_{t-1}\hspace{-1pt}\}\mathbb{P}^{g}\hspace{-3pt}\left\{\hspace{-1pt}x_t\hspace{-1pt},\hspace{-1pt}x_{t-1}\hspace{-1pt},\hspace{-1pt}h_{t-1}^{-i}\hspace{-1pt},\hspace{-1pt}a_{t-1}^{-i}|h_{t-1}^i\hspace{-1pt},\hspace{-1pt}a_{t-1}^i\hspace{-1pt},\hspace{-1pt}y_t^i\hspace{-1pt},\hspace{-1pt}z_t\hspace{-1pt}\right\}\nonumber\\
	=&
	\sum_{x_{t-1}}\hspace{-3pt}\Big[\mathbb{P}\{\hspace{-1pt}y_t^{-i}|x_{t},\hspace{-1pt}a_{t-1}\hspace{-1pt}\}\mathbb{P}\{\hspace{-1pt}x_t|x_{t-1},\hspace{-1pt}a_{t-1}\hspace{-1pt}\}
	\mathbb{P}^{g}	\left\{\hspace{-1pt}x_{t-1},\hspace{-1pt}h_{t-1}^{-i},\hspace{-1pt}a_{t-1}^{-i}|h_{t-1}^i,\hspace{-1pt}a_{t-1}^i,\hspace{-1pt}y_t^i,\hspace{-1pt}z_t\hspace{-1pt}\right\}\Big]\nonumber\\
	=&
	\sum_{x_{t-1}}\hspace{-3pt}\Big[\mathbb{P}\{\hspace{-1pt}y_t^{-i}|x_{t},a_{t-1}\hspace{-1pt}\}\mathbb{P}\{x_t|x_{t-1},a_{t-1}\}g^{-i}_{t-1}(h_{t-1}^{-i})(a_{t-1}^{-i})
	\mathbb{P}^{g}\hspace{-3pt}\left\{\hspace{-1pt}x_{t-1},h_{t-1}^{-i}|h_{t-1}^i,a_{t-1}^i,y_t^i,z_t\hspace{-1pt}\right\}\hspace*{-2pt}\Big]\nonumber\\
	=&
	\sum_{x_{t-1}}\hspace{-3pt}\Big[\mathbb{P}\{\hspace{-1pt}y_t^{-i}|x_{t},a_{t-1}\hspace{-1pt}\}\mathbb{P}\{x_t|x_{t-1},a_{t-1}\}g^{-i}_{t-1}(h_{t-1}^{-i})(a_{t-1}^{-i})	\frac{\mathbb{P}^{g}\hspace{-3pt}\left\{x_{t-1},h_{t-1}^{-i},y_t^i,z_t|h_{t-1}^i,a_{t-1}^i\right\}}{\mathbb{P}^{g}\hspace{-3pt}\left\{y_t^i,z_t|h_{t-1}^i,a_{t-1}^i\right\}}\hspace*{-2pt}\Big].\label{eq:thm1-1}
	\end{align} 
	Consider the term $\mathbb{P}^{g}\hspace{-3pt}\left\{x_{t-1},h_{t-1}^{-i},y_t^i,z_t|h_{t-1}^i,a_{t-1}^i\right\}$ in the nominator of the expression above. We have,
	\begin{align}
	&\mathbb{P}^{g}\hspace{-3pt}\left\{x_{t-1},h_{t-1}^{-i},y_t^i,z_t|h_{t-1}^i,a_{t-1}^i\right\}\nonumber\\
	=
	&\sum_{a_{t-1}^{-i},x_t}\hspace*{-5pt}\Big[\mathbb{P}\{y_t^i,z_t|x_t,a_{t-1}^{-i},a_{t-1}^i\}\mathbb{P}\{x_t|x_{t-1},a_{t-1}^{-i},a_{t-1}^i\}
	g^{-i}_{t-1}(h_{t-1}^{-i})(a_{t-1}^{-i})\mathbb{P}^{g}\hspace{-3pt}\left\{x_{t-1},h_{t-1}^{-i}|h_{t-1}^i,a_{t-1}^{i}\right\}\hspace*{-2pt}\Big]\nonumber\\
	=
	&\sum_{a_{t-1}^{-i},x_t}\hspace*{-5pt}\Big[\mathbb{P}\{y_t^i,z_t|x_t,a_{t-1}^{-i},a_{t-1}^i\}\mathbb{P}\{x_t|x_{t-1},a_{t-1}^{-i},a_{t-1}^i\}g^{-i}_{t-1}(h_{t-1}^{-i})(a_{t-1}^{-i})\mathbb{P}^{g^{-i}}\hspace{-3pt}\left\{x_{t-1},h_{t-1}^{-i}|h_{t-1}^i,a_{t-1}^{i}\right\}\hspace*{-2pt}\Big]\nonumber\\
	=
	&\mathbb{P}^{g^{-i}}\hspace{-3pt}\left\{x_{t-1},h_{t-1}^{-i},y_t^i,z_t|h_{t-1}^i,a_{t-1}^i\right\}\label{eq:thm1-2}
	\end{align}	
	where the second equality follows from the induction hypothesis (\ref{eq:beliefindependence1}) for $t\hspace*{-2pt}-\hspace*{-2pt}1$. Consequently, we also have,	
	\begin{gather}
	\mathbb{P}^{g}\hspace{-3pt}\left\{\hspace{-1pt}y_t^i,\hspace{-1pt}z_t|h_{t-1}^i,\hspace{-1pt}a_{t-1}^i\hspace{-2pt}\right\}\hspace{-2pt}
	=\hspace{-3pt}
	\sum_{\hat{h}_{t-1}^{-i},\hat{x}_t}\hspace{-5pt}\mathbb{P}^{g}\hspace{-3pt}\left\{y_t^i,z_t,\hat{x}_{t-1},\hat{h}_{t-1}^{-i}|h_{t-1}^i,a_{t-1}^i\hspace{-1pt}\right\}\nonumber\\
	\stackrel{\text{by (\ref{eq:thm1-2})}}{=}\nonumber\\
	\hspace{-8pt}\sum_{\hat{h}_{t\hspace{-0.5pt}-\hspace{-0.5pt}1}^{-i}\hspace{-1pt},\hspace{-0.5pt}\hat{x}_t}\hspace{-6pt}\mathbb{P}^{g^{\hspace{-1pt}-i}}\hspace{-1pt}\hspace{-3pt}\left\{\hspace{-2pt}y_t^i\hspace{-1pt},\hspace{-2pt}z_t\hspace{-1pt},\hspace{-2pt}\hat{x}_{t\hspace{-0.5pt}-\hspace{-0.5pt}1}\hspace{-1pt},\hspace{-2pt}\hat{h}_{t\hspace{-0.5pt}-\hspace{-0.5pt}1}^{-i}\hspace{-1pt}|\hspace{-0.5pt}h_{t\hspace{-0.5pt}-\hspace{-0.5pt}1}^i\hspace{-1.5pt},\hspace{-2pt}a_{t\hspace{-0.5pt}-\hspace{-0.5pt}1}^i\hspace{-2pt}\right\}\hspace{-3pt}=\hspace{-2pt}\mathbb{P}^{g^{\hspace{-1pt}-i}}\hspace{-4pt}\left\{\hspace{-1pt}y_t^i\hspace{-1pt},\hspace{-1pt}z_t\hspace{-1pt}|\hspace{-0.5pt}h_{t\hspace{-0.5pt}-\hspace{-0.5pt}1}^i\hspace{-1.5pt},\hspace{-2pt}a_{t\hspace{-0.5pt}-\hspace{-0.5pt}1}^i\hspace{-2pt}\right\}\hspace{-8pt}\label{eq:thm1-3}
	\end{gather}	
	Substituting (\ref{eq:thm1-2}) and (\ref{eq:thm1-3}) in (\ref{eq:thm1-1}),
	\begin{align}
	&\mathbb{P}^{g}\hspace{-3pt}\left\{\hspace{-1pt}x_t,\hspace{-1pt}h_t^{-i}|h_t^i\hspace{-1pt}\right\}\nonumber\\
	=&\sum_{x_{t-1}}\hspace{-3pt}\Big[\mathbb{P}\{\hspace{-1pt}y_t^{-i}|x_{t},a_{t-1}\hspace{-1pt}\}\mathbb{P}\{x_t|x_{t-1},a_{t-1}\}g^{-i}_{t-1}(h_{t-1}^{-i})(a_{t-1}^{-i})	\frac{\mathbb{P}^{g^{-i}}\hspace{-3pt}\left\{\hspace{-1pt}x_{t-1},\hspace{-1pt}h_{t-1}^{-i},\hspace{-1pt}y_t^i,\hspace{-1pt}z_t|h_{t-1}^i,\hspace{-1pt}a_{t-1}^i\right\}}{\mathbb{P}^{g^{-i}}\hspace{-3pt}\left\{y_t^i,z_t|h_{t-1}^i,a_{t-1}^i\right\}}\hspace*{-2pt}\Big]
	\nonumber\\
	=&\hspace{-2pt}\mathbb{P}^{g^{-i}}\hspace{-3pt}\left\{\hspace{-1pt}x_t,\hspace{-1pt}h_t^{-i}|h_t^i\hspace{-1pt}\right\}\nonumber
	\end{align}	
	which establishes the induction step  for $t$.

	Given the result of part (i), the result of part (ii) follows directly from the definition of SIB strategies (\ref{eq:SIB-strategy}) and SIB update rule (\ref{eq:CIBupdaterule}).
\end{proof}

To provide the proof for Theorem \ref{thm:nonstrategic}, we need the following result. \vspace*{-3pt}

\begin{lemma}
	\label{lemma:lem}
	Given a SIB strategy profile $\sigma$ and update rule $\psi$ consistent with $\sigma$,
	\begin{align}
	\hspace*{-3pt}\mathbb{P}^{\sigma}_\psi\{s_{t+1},\pi_{t+1}|p_t,c_t,a_t\}\hspace{-2pt}=\hspace{-2pt}\mathbb{P}^{\sigma}_\psi\{s_{t+1},\pi_{t+1}|s_t,\pi_t,a_t\},\hspace*{-2pt}	\end{align}
	for all $s_{t+1},\pi_{t+1},s_t,\pi_t,a_t$.
\end{lemma}

\vspace*{2pt}

\begin{proof}[\textbf{Proof of Lemma \ref{lemma:lem}}]
	Let $g^\sigma$ denote the strategy profile, given by (\ref{eq:SIB-strategy}), that corresponds to SIB strategy profile $\sigma$. We have,
	\begin{gather*}
	\mathbb{P}^{\sigma}_\psi\{s_{t+1},\hspace*{-1pt}\pi_{t+1}|p_t,\hspace*{-1pt}c_t,\hspace*{-1pt}a_t\}
	\stackrel{\pi_t=\gamma_t(c_t)}{=}
	\mathbb{P}^{\sigma}_\psi\{s_{t+1},\hspace*{-1pt}\pi_{t+1}|p_t,\hspace*{-1pt}c_t,\hspace*{-1pt}a_t,\hspace*{-1pt}\pi_t\}\\
	\stackrel{\text{using update rule (\ref{eq:CIBupdaterule})}}{=}\\
	\sum_{z_{t+1}}\Big[\mathbb{P}^{\sigma}_\psi\{s_{t+1},z_{t+1}|p_t,c_t,a_t,\pi_t\}\mathbf{1}_{\{\pi_{t+1}=\psi_{t+1}(\pi_t,z_{t+1})\}}\Big]\\
	\stackrel{\text{by  (\ref{eq:sufficientdynamic})}}{=}\\
	\hspace{-10pt}\sum_{z_{t+1}}\Big[\mathbb{P}^{\sigma}_\psi\{s_{t+1},z_{t+1}|s_t,c_t,a_t,\pi_t\}\mathbf{1}_{\{\pi_{t+1}=\psi_{t+1}(\pi_t,z_{t+1})\}}\Big]
	\\=\\	
	\sum_{y_{t+1},x_{t+1},x_t,z_{t+1}}\hspace{-25pt}\Big[\mathbb{P}^{\sigma}_\psi\{s_{t+1},z_{t+1},y_{t+1},x_{t+1},x_t|s_t,c_t,a_t,\pi_t\}\mathbf{1}_{\{\pi_{t+1}=\psi_{t+1}(\pi_t,z_{t+1})\}}\Big]\\
	\stackrel{\text{by system dynamics (\ref{eq:systemdynamic1}) and (\ref{eq:systemdynamic2})}}{=}\\
	\end{gather*}
	\begin{gather*}
	\hspace*{-120pt}\sum_{y_{t+1},x_{t+1},x_t,z_{t+1}}\hspace{-25pt}\Big[\mathbb{P}^{\sigma}_\psi\{s_{t+1}|s_t,c_t,a_t,\pi_t,z_{t+1},y_{t+1},x_{t+1},x_t\}\nonumber\\\hspace*{98pt}\mathbb{P}\{z_{t+1},y_{t+1}|a_t,x_{t+1}\}\mathbb{P}\{x_{t+1}|x_t,a_t\}
	\mathbb{P}^{\sigma}\{x_t|s_t,c_t,a_t,\pi_t\}\mathbf{1}_{\{\pi_{t+1}=\psi_{t+1}(\pi_t,z_{t+1})\}}\Big]\\
	\stackrel{\text{by  (\ref{eq:sufficientupdate})}}{=}\\
	\hspace*{-135pt}\sum_{y_{t+1},x_{t+1},x_t,z_{t+1}}\hspace{-5pt}\Big[\Big(\prod\limits_j\mathbf{1}_{\{s_{t+1}^j=\phi_{t+1}^j(s_t^j,\{y_{t+1}^j,z_{t+1},a_t^j\};g^\sigma)\}}\Big)\nonumber\\\hspace*{90pt}\mathbb{P}\{z_{t+1},y_{t+1}|a_t,x_{t+1}\}\mathbb{P}\{x_{t+1}|x_t,a_t\}
	\mathbb{P}^{\sigma}_\psi\{x_t|s_t,c_t\}\mathbf{1}_{\{\pi_{t+1}=\psi_{t+1}(\pi_t,z_{t+1})\}}\Big]\\
	\stackrel{\text{by Bayes' rule}}{=}\\
	\hspace*{-135pt}\sum_{y_{t+1},z_{t+1},x_{t+1},x_t}\hspace{-5pt}\Big[\Big(\prod\limits_j\mathbf{1}_{\{s_{t+1}^j=\phi_{t+1}^j(s_t^j,\{y_{t+1}^j,z_{t+1},a_t^j\};g^\sigma)\}}\Big)\\\hspace*{96pt}\mathbb{P}\{z_{t+1},y_{t+1}|a_t,x_{t+1}\}\mathbb{P}\{x_{t+1}|x_t,a_t\}\frac{\mathbb{P}^{\sigma}_\psi\{x_t,s_t|c_t\}}{\mathbb{P}^{\sigma}_\psi\{s_t|c_t\}}\mathbf{1}_{\{\pi_{t+1}=\psi_{t+1}(\pi_t,z_{t+1})\}}\Big]\\
	=\\
	\hspace*{-135pt}\sum_{y_{t+1},z_{t+1},x_{t+1},x_t}\hspace{-5pt}\Big[\Big(\prod\limits_j\mathbf{1}_{\{s_{t+1}^j=\phi_{t+1}^j(s_t^j,\{y_{t+1}^j,z_{t+1},a_t^j\};g^\sigma)\}}\Big)\\\hspace{96pt}\mathbb{P}\{z_{t+1},y_{t+1}|a_t,x_{t+1}\}\mathbb{P}\{x_{t+1}|x_t,a_t\}
	\frac{\pi_t(x_t,s_t)}{\sum_{\hat{x}_t}\pi_t(\hat{x}_t,s_t)}\mathbf{1}_{\{\pi_{t+1}=\psi_{t+1}(\pi_t,z_{t+1})\}}\Big]
	\\=\\
	\mathbb{P}^{\sigma}_\psi\{s_{t+1},\pi_{t+1}|s_t,\pi_t,a_t\}.
	\end{gather*}
\end{proof}

\begin{proof}[Proof of Theorem \ref{thm:nonstrategic}]
Consider an arbitrary strategy profile $g$. We prove the existence of SIB strategy profile that is equivalent to $g$ by construction. 

With some abuse of notation, let $\sigma^{i}(\Pi_t,S_t^i,\omega_t)$ denote agent $i$'s strategy using the public randomization device $\omega_t$. We construct a SIB strategy profile $\sigma_t$ that has the following properties:
\begin{enumerate}[(a)]
	\item  
	the induced distribution on $\{\Pi_{t+1},S_{t+1}\}$ under $\sigma$ coincides with one under $g$, \textit{i.e.}
	\begin{align}
	\mathbb{P}^{\sigma_{1:t}}\left\{\pi_{t+1},s_{t+1}\right\}=\mathbb{P}^{g_{1:t}}\left\{\pi_{t+1},s_{t+1}\right\}.\label{eq:proof-thm4-a}
	\end{align}
	\item 
	the continuation payoff for all the agents under $\sigma$ is the same as that under $g$, \textit{i.e.}  for all $i\in\mathcal{N}$,
	\begin{gather}
	\mathbb{E}^g\left\{\sum_{\tau=t}^{T}u_\tau^i(X_\tau,g_\tau(H_\tau))\right\}=\hspace*{5pt}\mathbb{E}^\sigma\left\{\sum_{\tau=t}^{T}u^{i}_\tau(X_\tau,\sigma_\tau(\Pi_\tau,S_\tau,\omega_\tau))\right\}.\label{eq:proof-thm4-b}
	\end{gather}
	
\end{enumerate}

We prove condition (a) by forward induction and condition (b) by backward induction. We note that condition (a) is satisfied for $t=1$, since at $t=1$ no action has been taken. Moreover, condition (b) is satisfied for $t=T+1$ since there is no future.
\vspace{5pt}

Assume that condition (a) is satisfied from $1$ to $t$, $t\hspace*{-2pt}\in\hspace*{-2pt}\mathcal{T}$.  We construct $\sigma_t$ below such that condition (a) is satisfied at $t\hspace*{-2pt}+\hspace*{-2pt}1$. 

To construct $\sigma_t$, we first define below a random vector $R_t^{1:N}$ based on $H_t^{1:N}$, such that for every $i\hspace{-2pt}\in\hspace{-2pt}\mathcal{N}$, (i) $R_t^{1:N}$ is independent of $\Pi_t$ and $S_t^{1:N}$, and (ii) $H_t^i$ can be reconstructed using $R_t^i$ along with $\Pi_t$ and $S_t^{i}$.

We proceed as follows. For every time $t\in\mathcal{T}$, let $(\pi_t,s_t^{1:N})$ denote the realization of the agents' sufficient common information and private information, respectively. Let $\mathcal{H}_t^i:=\{h_t^{i,1},...,h_t^{i,|\mathcal{H}_t^i|}\}$ denote the set of all histories of agent $i$ at time $t$, where $|\mathcal{H}_t^i|$ denote the number of possible realizations of agent $i$'s history at time $t$. Conditioned on the realization of $(\pi_t,s_t^i)$, let $\{p(h_t^{i,k}|\pi_t,s_t^i),1\hspace{-2pt}\leq\hspace{-2pt} k\hspace{-2pt}\leq\hspace{-2pt} |H_t^i|)\}$ denote the probability mass function on $\mathcal{H}_t^i$ that leads to $(\pi_t,s_t^i)$ for agent $i$. Define the random variable $R_t^i$ on $[0,1]$ as follows: \vspace*{3pt}
\begin{align}
\hspace*{-27pt}\text{1)}\hspace*{27pt}\mathbb{P}\Big\{\hspace*{-1pt}0\hspace*{-2pt}\leq \hspace*{-2pt}R_t^i \hspace*{-2pt}\leq\hspace*{-2pt} p(h_t^{i,1}|\pi_t,\hspace*{-1pt}s_t^i)\hspace*{-1pt}\Big\}\hspace*{-2pt}=\hspace*{-1pt}p(h_t^{i,1}|\pi_t,\hspace*{-1pt}s_t^i),\label{eq:thm2-r-1}
\end{align} 
and conditioned on the event $
\Big\{\hspace*{-1pt}0\hspace*{-3pt}\leq \hspace*{-3pt}R_t^i \hspace*{-3pt}\leq\hspace*{-3	pt} p(h_t^{i,1}\hspace*{-1pt}|\pi_t,\hspace*{-1pt}s_t^i)\hspace*{-1pt}\Big\}$, $R_t^i$ is uniformly distributed \hspace*{-3pt} on \hspace*{-1pt} $[0,\hspace*{-1pt}p(h_t^{i,1}|\pi_t,\hspace*{-1pt}s_t^i)]$. \vspace*{5pt}
	
2) For $1<k\leq |\mathcal{H}_t^i|$,
\begin{align} 
\hspace*{1pt}\mathbb{P}\Big\{\hspace*{-2pt}\sum_{j=1}^{k-1}p(h_t^{i,j}|\pi_t,\hspace*{-1pt}s_t^i)\hspace*{-2pt}\leq \hspace*{-2pt}R_t^i \hspace*{-2pt}\leq\hspace*{-2pt} \sum_{j=1}^{k}p(h_t^{i,j}|\pi_t,\hspace*{-1pt}s_t^i)\hspace*{-2pt}\Big\}\hspace*{-3pt}=\hspace*{-1pt}p(h_t^{i,k}|\pi_t,\hspace*{-1pt}s_t^i),\label{eq:thm2-r-2}
\end{align}	
and conditioned on the event $\Big\{\hspace*{-2pt}\sum_{j=1}^{k-1}p(h_t^{i,j}|\pi_t,\hspace*{-1pt}s_t^i)\hspace*{-2pt}\leq \hspace*{-2pt}R_t^i \hspace*{-2pt}\leq\hspace*{-2pt} \sum_{j=1}^{k}p(h_t^{i,j}|\pi_t,\hspace*{-1pt}s_t^i)\hspace*{-2pt}\Big\}$, $R_t^i$ is uniformly distributed \hspace*{-3pt} on \hspace*{-1pt} $\Big[\hspace*{-1pt}\sum_{j=1}^{k-1}p(h_t^{i,j}|\pi_t,\hspace*{-1pt}s_t^i)\hspace*{-2pt},\sum_{j=1}^{k}p(h_t^{i,j}|\pi_t,\hspace*{-1pt}s_t^i)\hspace*{-1pt}\Big]$. \vspace*{5pt}

Therefore, $R_t^i$ is uniformly distributed on $[0,1]$ and is independent of $(\Pi_t,S_t^i)$. Furthermore, for any realization $(\pi_t,s_t^i,r_t^i)$ 
we can uniquely determine $h_t^{i,l}$ where \vspace*{-2pt}
\begin{align}
l:=\min\{k:r_t^i\geq \sum_{j=1}^{k-1} p(h_t^{i,j}|\pi_t,s_t^i)\}. \label{eq:thm2-r-inv}\vspace*{-2pt}
\end{align}

 Therefore, the random variable $R_t^{i}$ defined above, satisfies the mentioned-above conditions (i) and (ii)
 when $H_t^{i}$ takes finite values.\vspace{5pt}
 
 We show below that $R_t^i$ is independent of $S_t$.

	\begin{lemma}\label{lemma:noise}
		The random variable $R_t^i$, $i\in\mathcal{N}$, is independent of $\Pi_t$ and $S_t$ for all $t\in\mathcal{T}$.
	\end{lemma}
   
   \begin{proof}[Proof of Lemma \ref{lemma:noise}] Consider an arbitrary realization $(\hspace*{-1pt}h_t^1\hspace*{-1pt},\hspace*{-1pt}...,\hspace*{-1pt}h_t^N\hspace*{-1pt})$ of $(\hspace*{-1pt}H_t^1\hspace*{-1pt},\hspace*{-1pt}...,\hspace*{-1pt}H_t^N\hspace*{-1pt})$. Let $((s_t^1\hspace*{-1pt},\hspace*{-1pt}\pi_t,r_t^1),\hspace*{-1pt}$ $...,\hspace*{-1pt}(s_t^N\hspace*{-1pt},\hspace*{-1pt}\pi_t,r_t^N))$ denote the realization of $((S_t^1,\Pi_t,R_t^1),...,(S_t^N,\Pi_t,R_t^N))$ where $(s_t^i,\pi_t,r_T^i)$ corresponds to $h_t^i$ as it is defined above for every $i\in\mathcal{N}$.  
   	
   For every $i\in\mathcal{N}$ we have,
   \begin{align}
   \mathbb{P}^g\{r_t^i|\pi_t,s_t\}&=\mathbb{P}^g\{r_t|\pi_t,s_t^i,s_T^{-i}\}=\frac{\mathbb{P}^g\{r_t^i,s_t^{-i}|\pi_t,s_t^i\}}{\mathbb{P}^g\{s_t^{-i}|\pi_t,s_t^i\}}
   =\frac{\mathbb{P}^g\{s_t^{-i}|r_t^i,\pi_t,s_t^i\}\mathbb{P}^g\{r_t^i|\pi_t,s_t^i\}}{\mathbb{P}^g\{s_t^{-i}|\pi_t,s_t^i\}}\nonumber\\&\hspace*{-29pt}\stackrel{\scriptsize \begin{array}{c}
   		\text{replace}  \\(\pi_t,s_t^i,r_t^i) \text{ by } h_t^i\end{array}}
   		{=}\frac{\mathbb{P}^g\{s_t^{-i}|h_t^i\}\mathbb{P}^g\{r_t^i|\pi_t,s_t^i\}}{\mathbb{P}^g\{s_t^{-i}|\pi_t,s_t^i\}}\label{eq:lemma5-0}
   \end{align}	
   The last equality holds because $H_t^i$ is uniquely determined by $(\Pi_t,S_t^i, R_t^i)$ and vice versa; see (\ref{eq:thm2-r-1})-(\ref{eq:thm2-r-inv}). Moreover,
   \begin{align}
   \mathbb{P}^g\{s_t^{-i}|h_t\}&\stackrel{by\; (\ref{eq:sufficientinfo})}{=}\mathbb{P}^g\{s_t^{-i}|s_t,c_t\}=\frac{\mathbb{P}^g\{s_t^{-i},s^i_t|c_t\}}{\mathbb{P}^g\{s_t^i|c_t\}}=\frac{\pi_t^g(s_t^{-i},s_t^i)}{\sum_{\hat{s}_t^{-i}}\pi_t^g(\hat{s}_t^{-i},s_t^i)}=\mathbb{P}\{s_t^{-i}|s_t,\pi_t^g\}.\label{eq:lemma5-1}
   \end{align}	
   
Combining (\ref{eq:lemma5-0}) and (\ref{eq:lemma5-1}) 
\begin{align}
\mathbb{P}^g\{r_t^i|\pi_t,s_t\}&=\frac{\mathbb{P}^g\{s_t^{-i}|s_t^i,\pi_t^g\}\mathbb{P}^g\{r_t^i|\pi_t,s_t^i\}}{\mathbb{P}^g\{s_t^{-i}|\pi_t,s_t^i\}}=\mathbb{P}^g\{r_t^i|\pi_t,s_t^i\}=\mathbb{P}^g\{r_t^i\}\label{eq:proof-thm4-indep}
\end{align}      	
where the last equality is true since by definition $R_t^i$ is independent of $(\Pi_t,S_t^i)$. Therefore, by (\ref{eq:proof-thm4-indep}), $R_t^i$ is independent of $\Pi_t$ and $S_t$ for all $i\in\mathcal{N}$. 
\end{proof}

Using the result of Lemma \ref{lemma:noise},  we prove that for every $i\hspace{-2pt}\in\hspace{-2pt}\mathcal{N}$, (i) $R_t^{1:N}$ is independent of $\Pi_t$ and $S_t^{1:N}$, and (ii) $H_t^i$ can be reconstructed using $R_t^i$ along with $\Pi_t$ and $S_t^{i}$. 

In the following, we construct a SIB strategy profile $\sigma_t$ equivalent to $g_t$ as follows. Let $\hat{R}_t^{1:N}(\omega_t)$ denote a random vector the agents construct using the public randomization device $\omega_t$ that has an identical joint cumulative distribution to that of $R_t^{1:N}$. Note that by Lemma \ref{lemma:noise}, the distribution of $R_t^{1:N}$ is independent of $S_t$ and $\Pi_t$.   

 Define,
\begin{align}
\sigma^{i}_t(\Pi_t,S_t^i,\omega_t):=g_t^i(F^{-1}_{R_t^i|S_t^i,\Pi_t}(\hat{R}_t^i(\omega_t),\Pi_t,S_t^i)).\label{eq:thm:team-1}
\end{align}

Then,
\begin{align*}
\mathbb{P}^{g_{1:t}}\{\pi_{t+1},s_{t+1}|H_t\}=&\mathbb{P}^{g_{1:t}}\{\pi_{t+1},s_{t+1}|\Pi_t,S_t,R_t\}\\\stackrel{distribution}{=}&\mathbb{P}^{g_{1:t}}\{\pi_{t+1},s_{t+1}|\Pi_t,S_t,\hat{R}_t\}\\=&\mathbb{P}^{\sigma_{1:t}}\{\pi_{t+1},s_{t+1}|\Pi_t,S_t,\hat{R}_t\}.
\end{align*}

Taking the expectation of the left and right hand sides with respect to $\omega_t$ and $R_t$, respectively, and using the fact that $\hat{R}_t(\omega_t)$ and $R_t$ are independent of $S_t$ and $\Pi_t$ (Lemma \ref{lemma:noise}), we obtain 
\begin{align}
&\mathbb{P}^{\sigma_{1:t}}\left\{\pi_{t+1},s_{t+1}|\Pi_t,S_t\right\}=\mathbb{P}^{g_{1:t}}\left\{\pi_{t+1},s_{t+1}|\Pi_t,S_t\right\} \quad w.p.1. \label{eq:thm4-induc(a)}
\end{align} 

By the induction hypothesis, we have $\mathbb{P}^{\sigma_{1:t-1}}\left\{\pi_{t},s_{t}\right\}=\mathbb{P}^{g_{1:t-1}}\left\{\pi_{t},s_{t}\right\}$. Therefore, taking the expectation of both sides of (\ref{eq:thm4-induc(a)}) with respect to $\Pi_t,S_t$, we establish that condition (a) holds for time $t+1$.

Next, we prove condition (b) by backward induction.  
We have,
\begin{align}
\mathbb{E}^g\{u_t^i(X_t,A_t)|H_t\}=&\mathbb{E}^g\{u_t^i(X_t,A_t)|\Pi_t,S_t,R_t\}\nonumber\\\stackrel{distribution}{=}&\mathbb{E}^g\{u_t^i(X_t,A_t)|\Pi_t,S_t,\hat{R}_t\}\nonumber\\=&\mathbb{E}^\sigma\{u_t^i(X_t,A_t)|\Pi_t,S_t,\hat{R}_t\} \label{eq:thm2-conb}
\end{align}

Using (\ref{eq:thm2-conb}) for $t=T$, we have condition (b) is satisfied for $t=T$.

Now we assume that condition (b) is satisfied from $t+1$ to $T$, $t\in\mathcal{T}$. We prove that condition (b) is satisfied at $t$.

Using condition (a) at time $t$, \textit{i.e} $\mathbb{P}^{\sigma_{1:t-1}}\{s_t,\pi_t\}=\mathbb{P}^{g_{1:t-1}}\{s_t,\pi_t\}$, the induction hypothesis on condition (b) for $t+1$ along with equation (\ref{eq:thm2-conb}) for $t$, and the fact that $R_t$ and $\hat{R}_t$ are identically distributed and independent of $\Pi_t$ and $S_t$, we obtain
\begin{gather*}
\mathbb{E}^g\hspace*{-2pt}\left\{\hspace*{-1pt}\sum_{\tau=t}^{T}\hspace*{-1pt}u_\tau^i(X_\tau,g_\tau(H_\tau)\hspace*{-1pt})\hspace*{-2pt}\right\}\hspace*{-2pt}=\hspace*{-1pt}\mathbb{E}^\sigma\hspace*{-2pt}\left\{\hspace*{-1pt}\sum_{\tau=t}^{T}\hspace*{-1pt}u^{i}_\tau(X_\tau,\sigma_\tau(\Pi_\tau,S_\tau,\omega_\tau)\hspace*{-1pt})\hspace*{-2pt}\right\}\hspace*{-2pt}.
\end{gather*}
\end{proof}

\vspace*{-10pt}

\begin{proof}[\textbf{Proof of Theorem \ref{thm:team}}]
	By the result of Theorem \ref{thm:nonstrategic}, we can restrict attention to SIB strategies with public randomization device without loss of generality. Moreover, since by Assumption \ref{assump:finite} all space are finite, we can restrict attention to SIB strategies (with no public randomization device) without loss of generality. The proof of Theorem \ref{thm:team} then follows from an argument identical to the one given for dynamic programming for POMDP (see \cite[Ch. 6.7]{kumar1986stochastic}). 
	
	The dynamic program described by (\ref{eq:team-initial}-\ref{eq:team-value}) can be viewed as a solution to the following decision problem that is equivalent to the original dynamic team problem. Consider a ``\textit{super agent}'' \red{that knows the functional forms of system dynamics and the agents' utilities, and the set of spaces $\mathcal{X}_t,\mathcal{A}_t^{1:N},\mathcal{S}_t^{1:N}$ for all $t$. The super agent coordinates the agents' decisions at each time as follows. The super agent observes $\pi_t$ (which is common knowledge among all agents) but does not know the realizations $s_t^{1:N}$ of the agents' sufficient private information.} Based on his information, the super agent chooses a joint set of prescriptions/partial functions $\sigma_t^{1:N}(\pi_t,\cdot)$, one for each agent, that determine agent $i$'s action for every realization $s_t^i$ as $\sigma^i_t(\pi_t,s_t^i)$ for $\forall t,i$.
	 The dynamic program described by (\ref{eq:team-initial}-\ref{eq:team-value}) determines an optimal solution for the above-described super agent, and thus, equivalently, determine the optimal strategy for the original dynamic team problem.\footnote{The above interpretation of the dynamic program from the point of view of a super agent is similar to \textit{the coordinator problem} formulated in \cite{nayyar2011optimal,nayyar2013decentralized}.} 
\end{proof}

\begin{proof}[\textbf{Proof of Theorem \ref{lemma:general}}]
	We show below that $L_t^i:=(S_t^{i},\Pi_t)$, $i\in\mathcal{N}$, $t\in\mathcal{T}$ satisfies conditions (i)-(iv) of Definition \ref{def:sufficient-general}.
	
	Condition (i) is satisfied since both $S_t^{1:N}$ and $\Pi_t$ can be updated recursively via update rules $\phi_t^{1:N}$ and $\psi_t$, respectively, for every $t\in\mathcal{T}$.
	
	Condition (ii) is satisfied by Lemma \ref{lemma:lem}.
	
	To prove condition (iii), we have
	\begin{align}
	\mathbb{P}\{x_t|c_t,s_t\}=\frac{\mathbb{P}\{x_t,s_t|c_t\}}{\sum_{\hat{x}_t}\mathbb{P}\{\hat{x}_t,s_t|c_t\}}=\mathbb{P}\{x_t|\pi_t,s_t\}.\label{eq:lemma-gen-eq}
	\end{align}
	Therefore, 

	\begin{gather*}
	\mathbb{E}^{\tilde{g}_{1:t-1}^{-i}}\hspace*{-2pt}\left\{\hspace*{-2pt}u_t^i(\hspace*{-1pt}X_t\hspace*{-1pt},\hspace*{-1pt}A_t\hspace*{-1pt})\hspace*{-1pt}\Big|c_t\hspace*{-1pt},\hspace*{-1pt}p_t^i\hspace*{-1pt},\hspace*{-1pt}a_t\hspace*{-2pt}\right\}\stackrel{\text{by  (\ref{eq:payoff-relevant2})}}{=}
	\mathbb{E}^{\tilde{g}_{1:t-1}^{-i}}\hspace*{-2pt}\left\{\hspace*{-2pt}u_t^i(\hspace*{-1pt}X_t\hspace*{-1pt},\hspace*{-1pt}A_t\hspace*{-1pt})\hspace*{-1pt}\Big|c_t\hspace*{-1pt},\hspace*{-1pt}s_t^i\hspace*{-1pt},\hspace*{-1pt}a_t\hspace*{-2pt}\right\}\nonumber\\
	=\nonumber\\
	\mathbb{E}^{\tilde{g}_{1:t-1}^{-i}}\Bigg\{\mathbb{E}^{\tilde{g}_{1:t-1}^{-i}}\hspace*{-2pt}\left\{\hspace*{-2pt}u_t^i(\hspace*{-1pt}X_t\hspace*{-1pt},\hspace*{-1pt}A_t\hspace*{-1pt})\hspace*{-1pt}\Big|X_t\hspace*{-1pt},\hspace*{-1pt}a_t\hspace*{-2pt}\right\}\Bigg| c_t,s_t^i,a_t\Bigg\}\nonumber\\
	\stackrel{\text{by  (\ref{eq:lemma-gen-eq})}}{=}\nonumber\\
	\mathbb{E}^{\tilde{g}_{1:t-1}^{-i}}\Bigg\{\mathbb{E}^{\tilde{g}_{1:t-1}^{-i}}\hspace*{-2pt}\left\{\hspace*{-2pt}u_t^i(\hspace*{-1pt}X_t\hspace*{-1pt},\hspace*{-1pt}A_t\hspace*{-1pt})\hspace*{-1pt}\Big|X_t\hspace*{-1pt},\hspace*{-1pt}a_t\hspace*{-2pt}\right\}\Bigg| \pi_t,s_t^i,a_t\Bigg\}
	\nonumber\\=\nonumber\\\mathbb{E}^{\tilde{g}_{1:t-1}^{-i}}\hspace*{-2pt}\left\{\hspace*{-2pt}u_t^i(\hspace*{-1pt}X_t\hspace*{-1pt},\hspace*{-1pt}A_t\hspace*{-1pt})\hspace*{-1pt}\Big|\pi_t\hspace*{-1pt},\hspace*{-1pt}p_t^i\hspace*{-1pt},\hspace*{-1pt}a_t\hspace*{-2pt}\right\}
	\end{gather*}
	
	Condition (iv) holds since,
	\begin{align*}
	\mathbb{P}^{\tilde{g}^{-i}_{1:t-1}\hspace*{-1pt},\tilde{g}_{1:t-1}^{i}}\hspace*{-3pt}\left\{\hspace*{-2pt}l_t^{-i}\hspace*{-1pt}\Big|p_t^i\hspace*{-1pt},\hspace*{-1pt}c_t\hspace*{-2pt}\right\}=&\hspace*{2pt}\mathbb{P}^{\tilde{g}^{-i}_{1:t-1}}\hspace*{-3pt}\left\{\hspace*{-2pt}s_t^{-i}\hspace*{-1pt}\Big|p_t^i\hspace*{-1pt},\hspace*{-1pt}c_t\hspace*{-2pt}\right\}
	\stackrel{\text{by  (\ref{eq:payoff-relevant2})}}{=}\hspace*{3pt}\mathbb{P}^{\tilde{g}^{-i}_{1:t-1}}\hspace*{-3pt}\left\{\hspace*{-2pt}s_t^{-i}\hspace*{-1pt}\Big|s_t^i\hspace*{-1pt},\hspace*{-1pt}c_t\hspace*{-2pt}\right\}
	=\frac{\mathbb{P}^{\tilde{g}^{-i}_{1:t-1}}\hspace*{-2pt}\left\{s_t|c_t\right\}}{\mathbb{P}^{\tilde{g}^{-i}_{1:t-1},g_{1:t-1}^{i}}\hspace*{-2pt}\left\{s_t^i|c_t\right\}}\nonumber\\
	\stackrel{\text{by  (\ref{eq:lemma-gen-eq})}}{=}\hspace*{-7pt}&\hspace*{10pt}\mathbb{P}^{\tilde{g}^{-i}_{1:t-1}}\hspace*{-2pt}\left\{\hspace*{-2pt}s_t^{-i}\hspace*{-1pt}\Big|s_t^i\hspace*{-1pt},\hspace*{-1pt}\pi_t\hspace*{-2pt}\right\}
	=\mathbb{P}^{\tilde{g}^{-i}_{1:t-1}}\hspace*{-2pt}\left\{\hspace*{-2pt}l_t^{-i}\hspace*{-1pt}\Big|l_t^i\hspace*{-1pt}\hspace*{-2pt}\right\}\nonumber
	\end{align*}
	
\end{proof}

\begin{proof}[\textbf{Proof of Theorem \ref{thm:nonstrategic-infinite}}]
Consider the SIB strategy $\sigma_t$ constructed in the proof of Theorem \ref{thm:nonstrategic} for every $t\in\mathbb{N}$. We show below that $\sigma_t$ satisfies (\ref{eq:nonstrategic-infinite}).  

By the proof of Theorem \ref{thm:nonstrategic}, condition (\ref{eq:proof-thm4-a}) holds for all $t\in\mathbb{N}$. To prove (\ref{eq:nonstrategic-infinite}), we show that under strategy $\sigma_t$, $t\in\mathbb{N}$, we have 
\begin{gather}
\Bigg|\mathbb{E}^g\hspace*{-2pt}\left\{\hspace*{-2pt}\sum_{\tau=t}^{\infty}\hspace*{-1pt}\delta^{t-1}u_{\infty}^i(g_\tau^{1:N}\hspace*{-1pt}(\hspace*{-1pt}H_\tau^{1:N}\hspace*{-1pt}),\hspace*{-1pt}X_\tau\hspace*{-1pt})\hspace*{-2pt}\right\}- 
\mathbb{E}^{\sigma_{\infty}}\hspace*{-2pt}\left\{\hspace*{-2pt}\sum_{\tau=t}^{\infty}\hspace*{-1pt}\delta^{t-1}u^{i}_{\infty}\hspace*{-1pt}(\sigma_{\infty}^{1:N}\hspace*{-1pt}(\hspace*{-1pt}\Pi_\tau\hspace*{-1pt},\hspace*{-1pt}S_\tau^{1:N}),\hspace*{-2pt}X_\tau\hspace*{-1pt})\hspace*{-1pt}\right\}\Bigg|\leq \epsilon \label{eq:proof-infinite}
\end{gather} 
for all $\epsilon>0$.

Let $M=\max_{a_t,x_t,i}|u^i_\infty(x_t,a_t)|$. For every $\epsilon>0$, choose $T\in\mathbb{N}$ such that $\frac{\delta^T}{1-\delta}M\leq \frac{\epsilon}{2}$. Then, for any arbitrary strategy $\tilde{g}$,
\begin{align}
\Bigg|\mathbb{E}^{\tilde{g}}\hspace*{-2pt}\left\{\hspace*{-2pt}\sum_{\tau=T}^{\infty}\hspace*{-1pt}\delta^{t-1}u_{\infty}^i(\tilde{g}_\tau^{1:N}\hspace*{-1pt}(\hspace*{-1pt}H_\tau^{1:N}\hspace*{-1pt}),\hspace*{-1pt}X_\tau\hspace*{-1pt})\hspace*{-2pt}\right\}\Bigg| \leq \frac{\epsilon}{2}. \label{eq:proof-infinite-2}
\end{align}

Therefore, for every $t<T$, condition (\ref{eq:proof-infinite}) is satisfied by (\ref{eq:proof-infinite-2}) and \textit{the triangle inequality}. 

For $t>T$, consider a finite decision problem with horizon $T$ resulting by the truncation of the original infinite-horizon decision problem at $T$. Then, by Theorem \ref{thm:nonstrategic}, 
\begin{gather}
\mathbb{E}^g\hspace*{-2pt}\left\{\hspace*{-2pt}\sum_{\tau=t}^{T}\hspace*{-1pt}u_\tau^i(g_\tau^{1:N}\hspace*{-1pt}(\hspace*{-1pt}H_\tau^{1:N}\hspace*{-1pt}),\hspace*{-1pt}X_\tau\hspace*{-1pt})\hspace*{-2pt}\right\}=
\mathbb{E}^\sigma\hspace*{-2pt}\left\{\hspace*{-2pt}\sum_{\tau=t}^{T}\hspace*{-1pt}u^{1:N}_\tau\hspace*{-1pt}(\sigma_\tau^{i}\hspace*{-1pt}(\hspace*{-1pt}\Pi_\tau\hspace*{-1pt},\hspace*{-1pt}S_\tau^{1:N}),\hspace*{-2pt}X_\tau\hspace*{-1pt})\hspace*{-1pt}\right\}, \label{eq:proof-infinite-3}
\end{gather}
for all $i\hspace*{-2pt}\in\hspace*{-2pt}\mathcal{N}$ and $t\hspace*{-2pt}\in\hspace*{-2pt}\mathcal{T}$. Combining (\ref{eq:proof-infinite-3}) with the result for $t>T$, we show that (\ref{eq:nonstrategic-infinite}) is satisfied for $t$. 
\end{proof}

\begin{proof}[\textbf{Proof of Theorem \ref{thm:team-inifnite}}]
	By Theorem \ref{thm:nonstrategic-infinite}, we can restrict attention to stationary SIB strategies with public randomization device without loss of generality. Moreover, since by Assumption \ref{assump:finite} all space are finite, we can restrict attention to SIB strategies (with no public randomization device) without loss of optimality. Consequently, following the same rationale as the one given in the proof of Theorem \ref{thm:team}, the result of Theorem \ref{thm:team-inifnite} follows from an argument identical to the one given for dynamic programming in infinite-horizon Markovian Decision Processes (see \cite[Ch. 8.2 and Ch.8.3]{kumar1986stochastic}).  
\end{proof}

	\end{document}